\documentclass[11pt]{amsart}
\pdfoutput=1

\usepackage{amssymb}
\usepackage{amsthm}
\usepackage{amsfonts}
\usepackage{amsmath}
\usepackage{pmboxdraw}
\usepackage{verbatim} 
\usepackage{graphicx}
\usepackage{color}
\usepackage[colorlinks=true, citecolor=blue, filecolor=black, linkcolor=black, urlcolor=black]{hyperref}
\usepackage{cite}
\usepackage[normalem]{ulem}
\usepackage{subcaption}
\usepackage{bbm}
\usepackage{bm}
\usepackage{mathtools}
\mathtoolsset{showonlyrefs}
\usepackage{todonotes}
\usepackage{enumitem}

\usepackage{xparse}

\ExplSyntaxOn
\NewDocumentCommand{\MeijerG}{smmmm}
 {
  \IfBooleanTF{#1}
   {
    \vic_meijerg:nnnnnn { #2 } { #3 } { #4 } { #5 } { small } { }
   }
   {
    \vic_meijerg:nnnnnn { #2 } { #3 } { #4 } { #5 } { } { \; }
   }
 }

\seq_new:N \l__vic_meijerg_args_in_seq
\seq_new:N \l__vic_meijerg_args_out_seq

\cs_new_protected:Nn \vic_meijerg:nnnnnn
 {
  \seq_set_split:Nnn \l__vic_meijerg_args_in_seq { | } { #3 }
  \seq_clear:N \l__vic_meijerg_args_out_seq
  \seq_map_inline:Nn \l__vic_meijerg_args_in_seq
   {
    \seq_put_right:Nn \l__vic_meijerg_args_out_seq
     {
      \begin{#5matrix} ##1 \end{#5matrix}
     }
   }
  G\sp{#1}\sb{#2}
  \left(
  \seq_use:Nn \l__vic_meijerg_args_out_seq { #6\middle|#6 }
  #6\middle|#6
  #4
  \right)
 }
\ExplSyntaxOff

\newcommand{\RN}[1]{%
	\textup{\uppercase\expandafter{\romannumeral#1}}%
}

\def\bp{{\bar\partial}}

\def\bfs{\boldsymbol}

\def\pa{\partial}

\DeclareMathOperator{\Res}{Res}

\def\NN{\mathcal{N}}

\def\SS{\mathcal{S}}
\def\TT{\mathcal{T}}

\def\C{\mathbb{C}}

\def\P{\mathbf{P}}
\def\R{\mathbb{R}}

\newcommand{\Pf}{{\textup{Pf}}}
\newcommand{\erfc}{\operatorname{erfc}}
\newcommand{\erf}{\operatorname{erf}}

\newcommand{\Prob}{{\mathbb{P}}}

\theoremstyle{plain}
\newtheorem*{thm*}{Theorem}
\newtheorem{thm}{Theorem}[section]

\newtheorem{cor}[thm]{Corollary}
\newtheorem{prop}[thm]{Proposition}
\newtheorem*{prop*}{Proposition}
\newtheorem*{lem*}{Lemma}

\theoremstyle{definition}
\newtheorem*{eg*}{Example}
\newtheorem*{egs*}{Examples}
\newtheorem{defi}{Definition}[section]
\newtheorem*{Q*}{Question}

\theoremstyle{remark}
\newtheorem*{rmk*}{Remark}
\newtheorem*{rmks*}{Remarks}

\newcommand{\abs}[1]{\lvert#1\rvert}
\numberwithin{equation}{section}

\begin{document}
\title[Universal number variance in 2D Coulomb gases]{Universality of the number variance in rotational invariant two-dimensional Coulomb gases}
\author{Gernot Akemann}
\address{Faculty of Physics, Bielefeld University, P.O. Box 100131, 33501 Bielefeld, Germany}
\email{akemann@physik.uni-bielefeld.de}

\author{Sung-Soo Byun}
\address{Center for Mathematical Challenges, Korea Institute for Advanced Study, 85 Hoegiro, Dongdaemun-gu, Seoul 02455, Republic of Korea}
\email{sungsoobyun@kias.re.kr}

\author{Markus Ebke}
\address{Department of Mathematics, Friedrich-Alexander-Universit\"at Erlangen-N\"urnberg, Cauerstrasse 11, 91058 Erlangen, Germany}
\email{markus.ebke@fau.de}


\date{\today}

\thanks{The work of Gernot Akemann was partly funded by the Deutsche Forschungsgemeinschaft (DFG) grant SFB 1283/2 2021 -- 317210226.
Sung-Soo Byun was partially supported by Samsung Science and Technology Foundation (SSTF-BA1401-51), by the National Research Foundation of Korea (NRF-2019R1A5A1028324) and by a KIAS Individual Grant (SP083201) via the Center for Mathematical Challenges at Korea Institute for Advanced Study.
Markus Ebke was partially supported by the DFG grant IRTG 2235.
}

\begin{abstract}
An exact map was established by Lacroix-A-Chez-Toine, Majumdar, and Schehr in \cite{lacroix2019rotating} between the $N$ complex eigenvalues of complex non-Hermitian  random matrices from the Ginibre ensemble, and the positions of $N$ non-interacting Fermions in a rotating trap in the ground state. An important quantity is the statistics of the number of Fermions $\mathcal{N}_a$ in a disc of radius $a$. Extending the work \cite{lacroix2019rotating} covering Gaussian and rotationally invariant potentials $Q$, we present a rigorous analysis in planar complex and symplectic ensembles, which both represent 2D Coulomb gases. We show that the variance of $\mathcal{N}_a$ is universal in the large-$N$ limit, when measured in units of the mean density proportional to $\Delta Q$, which itself is non-universal. This holds in the large-$N$ limit in the bulk and at the edge, when a finite fraction or almost all Fermions are inside the disc. In contrast, at the origin, when few eigenvalues are contained, it is the singularity of the potential that determines the universality class. We present three explicit examples from the Mittag-Leffler ensemble, products of Ginibre matrices, and truncated unitary random matrices. Our proofs exploit the integrable structure of the underlying determinantal respectively Pfaffian point processes and a simple representation of the variance in terms of truncated moments at finite-$N$.
\end{abstract}

\maketitle

\section{Introduction and discussion of main results}

There are many applications for random matrices with complex eigenvalues nowadays, for instance, in open quantum systems or quantum field theory with finite baryon chemical potential, and we refer to \cite{Fyodorov_2003} for a comprehensive review. Our main motivation comes from a recent physical realisation of the $N$ complex eigenvalues of the complex Ginibre ensemble which is Gaussian. In \cite{lacroix2019rotating} an exact map was constructed to the positions of $N$ non-interacting Fermions in two dimensions in the ground state, confined by a harmonic trap that is rotating with a particular frequency. It is also equivalent to $N$ free electrons in a perpendicular magnetic field in the lowest Landau level \cite[Section 15.2]{forrester2010log}. The first map was extended to higher order levels \cite{smith2021counting} and may allow for a direct comparison to experiments on cold atoms in such a setup, cf. \cite{Cooper_2008}. For a review of the extensive literature on the existence of such a map in one dimension to Hermitian ensembles, including finite temperature, we refer to \cite{Dean_2019}.

One of the key quantities computed in \cite{lacroix2019rotating}  was the expected number of Fermions in a disc of radius $a$ and its variance. We will investigate if the variance obtained there is universal, that is, if it holds for a much larger class of non-Gaussian ensembles. This question was already initiated in \cite{lacroix2019rotating} for more general rotationally invariant potentials, where also the entanglement entropy and all moments of the cumulant generating function were determined, which yields the variance as a special case. Three different limits were identified in \cite{lacroix2019rotating} where few, a finite fraction or almost all Fermions are contained in the disc. We will denote these by origin, bulk and edge limit, to be discussed in more detail below. We mention in passing that the intermediate and large deviations for the largest eigenvalue have been studied as well, see \cite{cunden2016large,lacroix2018extremes} and references therein.

The computation of the number of eigenvalues in a disc has a long history in non-Hermitian random matrix theory, including \cite{Grobe_1988,MR1181356,Mehta,jancovici1993large}. Perhaps surprisingly, it is much easier than in Hermitian ensembles to compute the probability that a disc is empty (gap probability), contains or exceeds (overcrowding) a prescribed number of eigenvalues, at least in the rotationally invariant case. This is because the corresponding eigenvalues of the Fredholm determinant (or Pfaffian) are explicitly known, and we refer to \cite{MR2536111} for a comprehensive study in complex and symplectic ensembles.
(See also \cite{charlier2021large,ghosh2018point,byun2022almost} and references therein for recent development in this direction.)
Other quantities have been analysed as well, including the number of eigenvalues in a generic, non-rotational invariant domain \cite{AR2016hole,adhikari2018hole,nishry2020forbidden} or the probability of overcrowding a domain \cite{akemann2013hole,akemann2014permanental,krishnapur2006overcrowding}.

Our goal is to extend the realm of universality of the number variance and test its limitations in a rigorous analysis. In the bulk limit, when a  finite fraction of eigenvalues (Fermions) are contained in a disc of radius $a$, the variance grows linearly with $a$ \cite{lacroix2019rotating}, when measured in units of the mean eigenvalue density at $a$ - which is itself a non-universal quantity. In the edge case, when almost all eigenvalues are contained, a universal scaling function
was given in \cite{lacroix2019rotating} for planar complex ensembles.
For the bulk and edge limit, we will show that in a completely different symmetry class, the symplectic Ginibre ensemble and its generalisation to rotationally invariant potentials, the same universal answer prevails. This is despite the breaking of rotational invariance on the level of the joint distribution of eigenvalues through the presence of complex conjugate eigenvalue pairs and a repulsion from the real axis \cite{ginibre1965statistical}. These planar symplectic ensembles also represent a two-dimensional Coulomb gas \cite{forrester2016analogies,kiessling1999note} and are interesting in their own right.

In contrast, in the origin limit the number variance becomes sensitive to the presence of singularities or zeros of the mean eigenvalue density. We will present three different ensembles with distinct behaviour at the origin: the Mittag-Leffler ensembles, products of $m$ Ginibre matrices, and truncated unitary random matrices. In all three cases we consider the symmetry classes with complex and with symplectic matrix elements.
The truncated ensembles are special as a limit of weak non-unitarity exists, where the eigenvalues converge to the unit circle. It was suggested in the limit of quantum systems with few open channels, see e.g.\ \cite{Fyodorov_2003} and references therein. At strong non-unitarity, we are back to the Ginibre universality class, including the origin limit as we will show.

A third ensemble of real non-symmetric Gaussian random matrices was introduced by Ginibre \cite{ginibre1965statistical}. While both the complex and symplectic Ginibre ensemble enjoy a simple integrable structure, yielding an almost identical result for the variance already at finite-$N$ as we will see, the Pfaffian structure of the real Ginibre ensemble is much more intricate. It separates into sectors with a different number of real eigenvalues, compare \cite{forrester2010log}. Thus we can only speculate if the same universal answer for the number variance found here in the bulk and at the edge also holds there.

\subsection{Basic setup} \label{Subsec_model intro}
We study the complex eigenvalues of two types of ensembles of random matrices with complex ($\beta=2$) or quaternion entries ($\beta=4$), whence the name planar ensembles. Furthermore, we will restrict ourselves to rotationally invariant potentials. In the Gaussian case these are represented by the complex respectively symplectic Ginibre ensemble. For higher order potentials with $\beta=2$, we consider random normal matrices.
We will label these two ensembles by the number of independent real degrees of freedom $\beta$ per matrix element. Although both correspond to a two-dimensional Coulomb gas \cite{forrester2010log}, $\beta$ does not represent the inverse temperature for $\beta=4$.

For the planar complex ensemble $(\beta=2)$ the joint probability distribution of points  $\{z_j\}_{j=1}^N\in\mathbb{C}^N$ is defined as
\begin{equation} \label{Gibbs cplx}
 d \P_N^{(2)} (z_1,\ldots, z_N)=\frac{1}{Z_N^{(2)}} \prod_{j>k} |z_j-z_k|^2 \prod_{j=1}^N e^{-NQ_N(z_j)}\, dA(z_j).
\end{equation}
Here $dA(z):=\frac{1}{\pi}d^2z$ is the normalised area measure, and the partition function reads
\begin{equation} \label{Z cplx}
 Z_N^{(2)} =\int_{\mathbb{C}^N}\prod_{j>k}^N |z_j-z_k|^2 \prod_{j=1}^N e^{-NQ_N(z_j)}\, dA(z_j).
\end{equation}
It contains the modulus square of the Vandermonde determinant, $\Delta_N(z_1,\ldots, z_N)=\prod_{j>k}^N (z_j-z_k)=\det_{1\leq j,k\leq N}[z_j^{k-1}]$.
The function $Q_N:\C \to \R$ is called external potential and satisfies suitable conditions to guarantee its existence, see Def.\ \ref{Qsuit} below.
In particular, if $Q_N$ is given by
\begin{equation} \label{Q Gin}
Q^{\textup{Gin}}(z):= \frac{\beta}{2} |z|^2,
\end{equation}
the ensemble \eqref{Gibbs cplx} reduces to the
complex
Ginibre ensemble \cite{ginibre1965statistical}.
Eq.~\eqref{Gibbs cplx} represents a determinantal point process,
see Section \ref{Var-finite-N} and Eq.~\eqref{RNk det} for more details later. All its correlation functions can be expressed in terms of the kernel of planar  orthogonal polynomials, which are monomials due to the rotational invariance here.

The map to the lowest Landau level of $N$ free electrons in a magnetic field can be seen as follows, and we refer to \cite{lacroix2019rotating} for more details. The single particle Hamiltonian in a two-dimensional harmonic trap with frequency $\omega$, rotating at frequency $\Omega$, is given by
\begin{equation}
\label{Hamiltonian}
H(\vec{p},\vec{r})=\frac{\vec{p}^2}{2m}+\frac{m\omega^2}{2}\vec{r}^2-\Omega\vec{e}_3\cdot \vec{r}\times\vec{p}=\frac{1}{2m}(\vec{p}-m\omega \vec{e}_3\times\vec{r})^2 +(\omega-\Omega)\vec{e}_3\cdot \vec{r}\times\vec{p}.
\end{equation}
Choosing units $\hbar=m=\omega=1$ and restricting ourselves to $0\leq \Omega\leq1$ to have bound states, the energy eigenvalues of the $N$ particle Hamilton operator $\sum_{i=1}^NH(\hat{\vec{p}}_i,\hat{\vec{r}}_i)$ are given by
\begin{equation}
E_{n_1,n_2}=1+(1-\Omega)n_1+(1+\Omega)n_2,
\end{equation}
with $n_1,n_2\in\mathbb{N}$. For $n_2=0$ we are in the lowest Landau level, which for $1-2/N<\Omega<1$ is non-degenerate. The $N$ lowest eigenfunctions $\phi_{n_1}$ are then proportional to
 \begin{equation}
 \phi_{n_1}(z)\sim z^{n_1} e^{-|z|^2/2},\qquad z=r_1+ir_2.
 \end{equation}
Consequently, the Slater determinant of the $N$ particle wave-function $\Psi_0(z_1,\ldots,z_N) \sim\Delta_N(z_1,\ldots,z_N)$ is proportionally to the Vandermonde determinant, and the amplitude $|\Psi_0|^2$ is the joint density \eqref{Gibbs cplx}.

For the planar symplectic ensemble $(\beta=4)$ the joint probability distribution of points $\{z_j\}_{j=1}^N$ has an additional complex conjugation symmetry, that is the eigenvalues come in complex conjugate pairs.
The point process follows the law
\begin{equation} \label{Gibbs symplectic}
 d \P_N^{(4)}(z_1,\ldots, z_N) =\frac{1}{Z_N^{(4)}} \prod_{j>k} |z_j-z_k|^2 |z_j-\bar{z}_k|^2 \prod_{j=1}^N |z_j-\bar{z}_j|^2 \, e^{-NQ_N(z_j)}\, dA(z_j),
\end{equation}
with
\begin{equation} \label{Z sympl}
 Z_N^{(4)} =\int_{\mathbb{C}^N}\prod_{j>k} |z_j-z_k|^2 |z_j-\bar{z}_k|^2 \prod_{j=1}^N |z_j-\bar{z}_j|^2 \prod_{j=1}^N e^{-NQ_N(z_j)}\, dA(z_j).
\end{equation}
Inserting \eqref{Q Gin} we obtain the symplectic Ginibre ensemble \cite{ginibre1965statistical}. Eq.~\eqref{Gibbs symplectic} represents a Pfaffian point process,
to be detailed in Section \ref{Var-finite-N} Eq.~\eqref{RNk Pfa} later. All its correlation functions can be expressed in terms of the pre-kernel of skew-orthogonal polynomials that will be introduced there.
Although the joint density \eqref{Gibbs symplectic} is proportional to the Vandermonde $\Delta_{2N}(z_1,\ldots, z_N, \bar{z}_1,\ldots,\bar{z}_N)$ \cite{MR1928853}, we are currently lacking a map to the Slater determinant of $N$ electrons in a rotating trap as in \eqref{Hamiltonian}.

Below we  consider a general radially symmetric potential for both classes of point processes,
\begin{equation}\label{Q radially sym}
Q_N(z) =g_N(|z|=r), \qquad g_N: \R_+ \to \R,
\end{equation}
which possibly depends on $N$.
Such an $N$-dependence is useful to describe ensembles that lead to different universality classes, see \eqref{g ML potential}, \eqref{Q products}, and \eqref{Q truncated unitary} below.
If the potential does not depend on $N$, we drop the subscript $N$ and simply write $Q\equiv Q_N$, $g \equiv g_N$.

Let us first discuss the macroscopic behaviour of these ensembles.
For this purpose, we denote $Q:=\lim_{N\to \infty} Q_N$ and $g:=\lim_{N\to \infty} g_N$, which is always assumed to be well defined.
Somewhat roughly speaking, the macroscopic behaviour of the ensembles depends only on $Q$.
It is well known \cite{HM13,MR2934715,kiessling1999note,chafai2014first} that the empirical measure converges to Frostman's equilibrium measure \cite{ST97} of the form
\begin{equation}\label{Frostman}
\rho(z)  =\frac{2}{\beta} \Delta Q(z) \cdot \mathbbm{1}_S(z),
\end{equation}
where $\Delta := \pa \bp=\frac14(\pa_r^2+\frac{1}{r} \pa_r)$, $S$ is a certain compact set called the droplet, and $\mathbbm{1}_S(z)$ is the indicator function on this set.

Let us formulate the following conditions.
\begin{defi}\label{Qsuit}
A rotational invariant potential $Q_N(z) =g_N(|z|)$ will be called \emph{suitable}, if it satisfies the following conditions:
\begin{enumerate}
\item $g_N(r) \gg \log r$ as $r \to \infty$,
\smallskip
\item $g_N \in C^2(0,1]$,
\smallskip
    \item $r g'_N(r)$ increases on $(0,\infty)$,
\smallskip
\item
$\lim_{r \to 0} r g'(r)=0$, and $g'(1)=\beta$.
\end{enumerate}
\end{defi}

The first condition is required to guarantee $Z_N^{(\beta)}<\infty.$
The second condition is made merely for convenience.
Due to the last condition
the support or droplet $S$ of the limiting spectral distribution is given by the unit disc, see \cite[Section \RN{4}.6]{ST97}, \cite[Section 2.7]{HM13} and \cite[Theorem 2.1]{MR2934715}.
For the Gaussian potential \eqref{Q Gin}, it corresponds to the circular law.
Note that
\begin{equation}
(rg'_N(r))'=g'_N(r)+rg''_N(r) = 4r \Delta Q_N(z).
\end{equation}
Thus $rg_N'(r)$ increases on $(0,\infty)$ if and only if $\Delta Q_N(z) > 0$ for $z \not= 0$, as formulated in the third condition.

For $a > 0$, let us write $\NN_a^{(\beta)}$ for the \emph{number of eigenvalues} in the disc of radius $a$ centered around the origin, $D_a=\{ z \in \C : |z|<a \}$.
Here, $\beta$ labels the respective ensemble.
We denote by
\begin{equation}\label{EN VN}
E_N^{(\beta)}(a):=\mathbb{E} \,\NN_a^{(\beta)}, \qquad  V_N^{(\beta)}(a):=\textup{Var }\NN_a^{(\beta)} ,
\end{equation}
the \emph{mean}  $E_N^{(\beta)}$, respectively the \emph{number variance}  $V_N^{(\beta)}$ of $\NN_a^{(\beta)}$.
It is easy to derive the leading order asymptotic behaviour of $E_N^{(\beta)}$, see Corollary \ref{Cor_in prob conv} and details in its proof, Eq.~\eqref{EN asymp}.
The primary purpose of this work is to derive closed formulas for $V_N^{(\beta)}$ in both symmetry classes at finite-$N$ and their asymptotic behaviour as $N \to \infty$.

\subsection{Number variance in the bulk and at the edge}\label{subsec:var}

In this subsection, we introduce our main results, where we begin
with finite-$N$.
For a general radially symmetric potential \eqref{Q radially sym}, let us write
\begin{equation} \label{ortho norm}
h_j:=\int_\C |z|^{2j} e^{-N Q_N(z)}\,dA(z) = 2 \int_0^\infty r^{2j+1} e^{-N g_N(r)}\,dr,
\end{equation}
for the non-vanishing $j$th moments, which are at the same time the squared norms of planar orthogonal polynomials (monomials).
We also write
\begin{align}
\begin{split} \label{ortho norm trunc}
h_{j,1}(a)&:= \int_{ |z|<a } e^{-NQ_N(z)} |z|^{2j}\,dA(z)=  2 \int_0^a e^{-N g_N(r)} r^{2j+1}  \,dr,
\\
h_{j,2}(a)&:= \int_{ |z|>a } e^{-NQ_N(z)} |z|^{2j}\,dA(z)=  2 \int_a^\infty e^{-N g_N(r)} r^{2j+1}  \,dr
\end{split}
\end{align}
for the truncated moments or squared norms.
Note that $h_j=h_{j,1}(a)+h_{j,2}(a).$
We then obtain the following.

\begin{prop} \textbf{\textup{(Mean and number variance of the number $\NN_a^{(\beta)}$  at finite-$N$)}}  \label{Prop_VN rep}
For each $N$, we have
\begin{equation} \label{EN 2 4 rep}
E_N^{(2)}(a)= \sum_{j=0}^{N-1} \frac{h_{j,1}(a)}{h_j}, \qquad
E_N^{(4)}(a) =  \sum_{j=0}^{N-1} \frac{ h_{2j+1,1}(a) }{ h_{2j+1} },
\end{equation}
and
\begin{gather}
\label{VN 2 4 rep}
V_N^{(2)}(a)= \sum_{j=0}^{N-1} \frac{h_{j,1}(a)\,h_{j,2}(a)}{h_j^2}, \qquad
V_N^{(4)}(a) =  \sum_{j=0}^{N-1} \frac{ h_{2j+1,1}(a) \, h_{2j+1,2}(a)  }{ h_{2j+1}^2 }.
\end{gather}
\end{prop}
This proposition plays a key role in proving our main result, Theorem~\ref{Thm_number variance} below. Notice that for $\beta=4$ only the odd moments appear.
Let us emphasise that for $\beta=2$, the expressions \eqref{EN 2 4 rep} and \eqref{VN 2 4 rep} also follow from \cite{lacroix2019intermediate} using the cumulant generating function.

We turn to our main result for the number variance in different large-$N$ limits.

\begin{thm} \label{Thm_number variance}
Let $Q_N$ be a suitable  potential according to Definition \ref{Qsuit}.
Then the following holds.
\begin{enumerate}[label=(\roman*)]
    \item \textup{\textbf{(Bulk)}} For $a\in (0,1)$ fixed, we have
\begin{equation} \label{VN bulk}
\lim_{N \to \infty} \frac{\beta }{ \sqrt{N \Delta Q(a)} }  V_N^{(\beta)}(a) =  \frac{2a}{\sqrt{\pi}} .
\end{equation}
\item \textup{\textbf{(Edge)}} For $\SS \in \R$, we have
\begin{equation} \label{VN edge}
\lim_{N \to \infty} \frac{ \beta }{ \sqrt{N \Delta Q(1) } } V_N^{(\beta)}\Big( 1-\frac{ \SS }{ \sqrt{ 2 \Delta Q(1)  N} } \Big) = \frac{2}{\sqrt{\pi}} \, f(\SS),
\end{equation}
where
\begin{equation}  \label{f(SS)}
f(\SS)  :=  \sqrt{ 2\pi }   \int_{-\infty}^{\SS} \frac{\erfc(t)\erfc(-t)}{4}  \,dt=  \frac{\erfc(-\sqrt{2}\SS)}{2} -\frac{ e^{-\SS^2} }{ \sqrt{2} } \erf(\SS) + \sqrt{ \frac{\pi}{2} } \, \SS \,\frac{\erfc(\SS)\erfc(-\SS)}{2} .
\end{equation}
\end{enumerate}
\end{thm}

Note that $f(\SS) \to 1$ as $\SS \to \infty$.
Thus one can see that as $\SS \to \infty$, the right-hand side of \eqref{VN edge} recovers that of \eqref{VN bulk} in the limit  $a\to1$.
We also remark that the Laplacian of $Q$  in \eqref{VN bulk} and \eqref{VN edge} corresponds to the macroscopic density $\rho$ in the large-$N$ limit, see \eqref{Frostman}. While $\rho$ itself is clearly non-universal, as it depends explicitly on the potential $Q$, the right hand side of the variance measured in units $N\Delta Q(a)$ is universal. To the best of our knowledge, Theorem~\ref{Thm_number variance} for $\beta=4$ has not appeared in the literature before.

The number variance of the complex Ginibre ensemble at $\beta=2$ was computed in \cite{lacroix2019rotating} in these limits, and was argued to be universal for more general rotational invariant potentials \cite{lacroix2018extremes}, using a saddle point approximation.
In \cite{lacroix2019rotating} the limit (i) we call ``bulk'' in Theorem \ref{Thm_number variance} was called ``extended bulk''. It is a global quantity, as the disc of radius $a$ contains a macroscopic, finite fraction of the eigenvalues. The edge regime (ii) (with the same name as in \cite{lacroix2019rotating}) is a local quantity, as we zoom into the vicinity of the edge of the limiting support, normalised to the unit disc (as mentioned after Definition \ref{Qsuit}).  A third limit called ``deep bulk'' was investigated in \cite{lacroix2019rotating} for the complex Ginibre ensemble. It is again a local quantity, as the vicinity of the origin is zoomed into. Although the origin is not special in the complex Ginibre ensemble, the corresponding variance shows a nontrivial behaviour, interpolating between quadratic and the linear behaviour in $a$ in the (extended) bulk.
We will also investigate this local origin limit in the next Subsection \ref{Sub:origin}, allowing for more general classes of potentials with singular behaviour at the origin.
These will give rise to different universality classes.

As a consequence of Theorem~\ref{Thm_number variance}, we obtain the convergence of the random variable $\mathcal{N}_a$ in probability.

\begin{cor}\label{Cor_in prob conv}
Under the same assumptions of Theorem~\ref{Thm_number variance}, we have that for any $a \in (0,1]$,
\begin{equation}\label{in prob conv}
\frac{1}{N} \mathcal{N}_a \to \frac{a g'(a)}{\beta}
\end{equation}
as $N\to\infty$, in probability.
\end{cor}

\subsection{Number variance at the origin}\label{Sub:origin}

In this subsection we present  the asymptotic behaviour of $V_N^{(\beta)}(a)$ near the origin $a = 0$.  This is achieved by rescaling $a={\TT}/N^\delta$ with some power $\delta>1$.
It depends on the local nature of the ensemble at the origin, and we shall present two related examples for ensembles exhibiting a singular or vanishing behaviour.

\medskip

\begin{eg*}[Mittag-Leffler ensemble]
This ensemble owes its name to the appearance of the two-parameter Mittag-Leffler function in its limiting kernel \cite{ameur2018random}.  It is defined through the potential
\begin{equation} \label{g ML potential}
Q^{\textup{ML}}_N(z):=\frac{\beta}{2b}|z|^{2b}-\frac{2c}{N}\log|z|,  \qquad b>0, \quad c>-1,
\end{equation}
which is suitable according to Definition \ref{Qsuit}.
Such a model was studied in \cite{chau1998structure,charlier2021large} for $\beta=2$ and \cite{akemann2021scaling} for $\beta=4$.
In particular, if $b=1$, the model is known as induced Ginibre ensemble, cf.  \cite{MR2881072}. It reduces to the standard Ginibre ensembles when also setting $c=0$.
The macroscopic density of the Mittag-Leffler ensemble is given by
\begin{equation} \label{ML macro density}
\frac{2}{\beta}\Delta Q^{\textup{ML}}(z) \cdot \mathbbm{1}_{\{ |z|<1 \} }= b |z|^{2b-2} \cdot \mathbbm{1}_{\{ |z|<1 \} }.
\end{equation}
Thus for $b <1(>1)$, the density reveals a singular (vanishing) behaviour at the origin on a macroscopic scale.
Furthermore, the insertion of a point charge $c$ at the origin in \eqref{g ML potential} gives rise to a conical type singularity on a local scale.

When $a$ is away from the origin, the number variance  follows as a specific case of Theorem~\ref{Thm_number variance}.
For $\beta=2$ this result was obtained by Charlier in \cite[Corollary 1.4]{charlier2022asymptotics} in an expansion including higher order correction terms. (See \cite{fenzl2022precise} for an earlier work on the complex Ginibre ensemble.) Indeed, \cite{charlier2022asymptotics} contains precise large-$N$ expansions of all higher cumulants of $\mathcal{N}_a$ for the complex Mittag-Leffler ensemble.
We also refer to \cite{byun2022characteristic,CL22} for a generalisation involving circular-root and merging type singularities.

Recall that the (regularised) incomplete gamma functions $P(a,z)$ and $Q(a,z)$ are given by
\begin{equation}
P(\alpha,z)=\frac{\gamma(\alpha,z)}{\Gamma(\alpha)}, \qquad Q(\alpha,z)=1-P(\alpha,z), \mbox{ with}\quad \gamma(a,z)=\int_0^z t^{a-1}e^{-t}dt,
\end{equation}
see e.g.\ \cite[Chapter 8]{olver2010nist}.

\begin{prop} \label{Prop_number variance ML ensemble}
\textup{\textbf{(Expected number and variance of the Mittag-Leffler ensemble at the origin)}}
Let $Q_N=Q_N^{\textup{ML}}$. After rescaling $a=\TT/N^{\frac{1}{2b}}$ with  $\TT>0$ fixed, we have
\begin{align}
\label{EN ML micro}
E^{(\beta)}( \TT)&=\lim_{N \to \infty} E_N^{(\beta)}\Big( \frac{\TT}{N^{ \frac{1}{2b} }} \Big)=  \sum_{j=1}^{\infty} P\Big( \frac{\beta j+2c}{2b}, \frac{\beta}{2b}\TT^{2b}  \Big) ,\\
 \label{VN ML micro}
V^{(\beta)}(\TT)&=\lim_{N \to \infty} V_N^{(\beta)}\Big( \frac{\TT}{N^{ \frac{1}{2b} }} \Big)=\sum_{j=1}^{\infty} P\Big( \frac{\beta j+2c}{2b}, \frac{\beta}{2b}\TT^{2b}  \Big) Q\Big( \frac{\beta j+2c}{2b}, \frac{\beta}{2b} \TT^{2b} \Big).
\end{align}
\end{prop}
When setting $b=1$ and $c=0$, we obtain the result for the Ginibre ensembles, reproducing the findings of \cite{lacroix2019rotating} for $\beta=2$ in what was called the deep bulk there. As it was also discussed there, let us look at the small argument limit $\TT\to0$ of \eqref{EN ML micro} and \eqref{VN ML micro}. It is not difficult to see that using the expansion \cite[Eq.(8.7.1)]{olver2010nist}
\begin{equation}
P(\alpha,z) = \frac{1}{\Gamma(\alpha)}\sum_{k=0}^\infty \frac{(-1)^k z^{k+\alpha}}{k!(\alpha+k)},
\end{equation}
together with $Q(\alpha,z)=1-P(\alpha,z)$, we obtain for both quantities for small $\TT$
\begin{equation}\label{EVMLsmallT}
E^{(\beta)}( \TT), V^{(\beta)}( \TT)\sim\frac{(\frac{\beta}{2b})^{(\beta+2c)/(2b)}}{\Gamma( \frac{\beta+2c}{2b}+1)} \TT^{\beta+2c}.
\end{equation}
The leading order contribution comes from the summand at $j=1$ and in fact holds already at finite-$N$. The fact that mean and variance agree indicates an underlying Poisson distribution - if we could show this for all cumulants, as it was done for the Ginibre ensemble at $\beta=2$ in \cite{lacroix2019rotating}.
The large argument limit  $\TT$ should match the linear behaviour found in Theorem \ref{Thm_number variance} in the bulk limit as $Q^{\rm ML}$ satisfies the suitability conditions.
\end{eg*}

\medskip

\begin{eg*}[Product ensemble] When considering the product of $m \in \mathbb{N}_+$ complex or quaternionic Ginibre matrices, the complex eigenvalues of the product matrix form a determinantal respectively Pfaffian point process , see \cite{MR2993423} for $\beta=2$ and \cite{MR3066113} for $\beta=4$. The corresponding
 potential is given by a Meijer $G$-function
\begin{equation} \label{Q products}
Q_N^{(m)}(z) := -\frac{1}{N} \log \MeijerG{m , 0}{0,  m}{ -  \\  \bfs{0}}{ \Big(\frac{\beta N}{2}\Big)^m |z|^2 },
\end{equation}
see  e.g.\ \cite[Chapter 16]{olver2010nist} and Appendix \ref{Appendix_Meijer G} for its definition. Here $\bfs{0}=(0,\ldots,0)$ denotes a row vector of length $m$.
This potential satisfies the suitability conditions in Definition \ref{Qsuit}.
For instance, if $m=1$, it follows from
\begin{equation} \label{MeijG 1001}
	\MeijerG {1,0} {0 ,1} { - \\ 0 } {x^2}= e^{-x^2}
\end{equation}
that $Q_N^{(m=1)}(z)=Q^{\textup{Gin}}(z)$. More general products have been studied, e.g.\ rectangular or truncated unitary random matrices, see \cite{Akemann_2015} for a review.

\begin{figure}[ht]
	\begin{subfigure}{0.31\textwidth}
		\begin{center}
			\includegraphics[width=\textwidth]{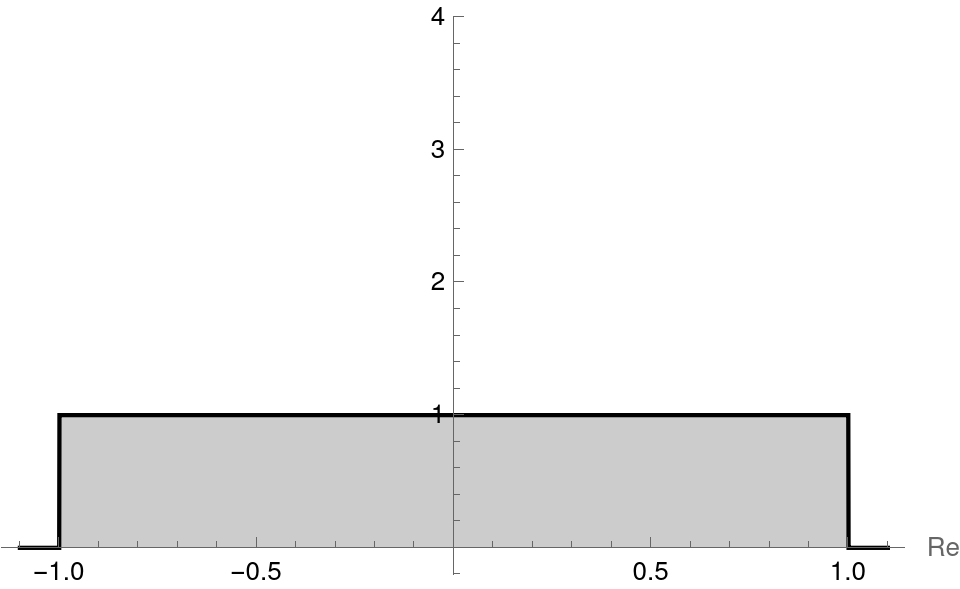}
		\end{center}
		\subcaption{$m=1$}
	\end{subfigure}
	\begin{subfigure}[h]{0.31\textwidth}
		\begin{center}
			\includegraphics[width=\textwidth]{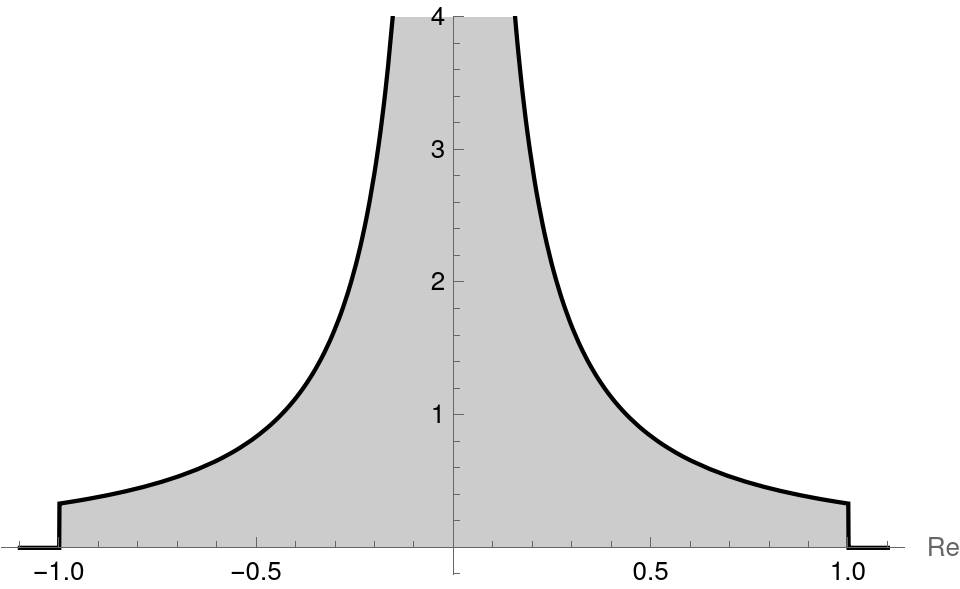}
		\end{center} \subcaption{$m=3$}
	\end{subfigure}
		\begin{subfigure}{0.31\textwidth}
		\begin{center}
			\includegraphics[width=\textwidth]{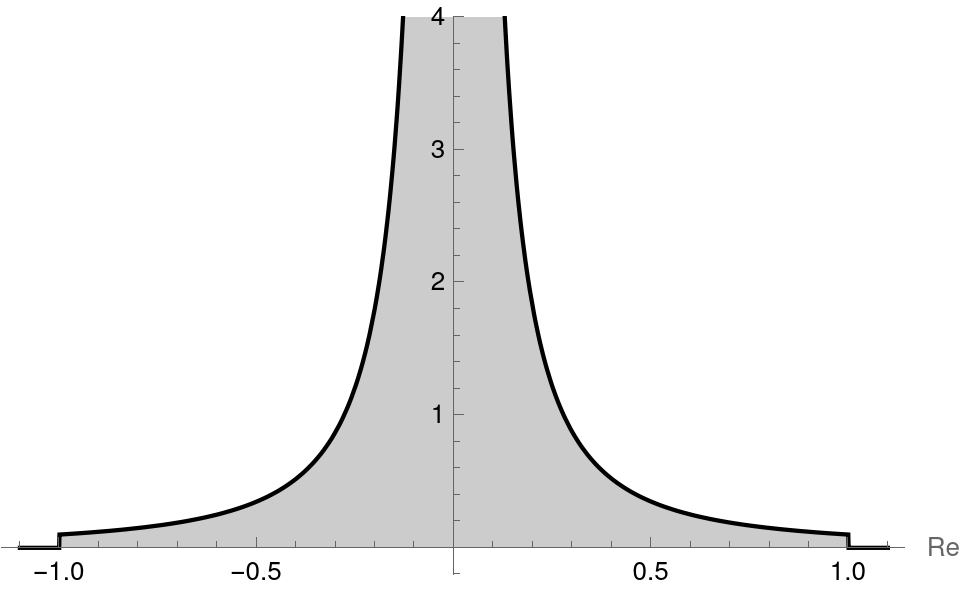}
		\end{center}
		\subcaption{$m=10$}
	\end{subfigure}

	\begin{subfigure}{0.31\textwidth}
		\begin{center}
			\includegraphics[width=\textwidth]{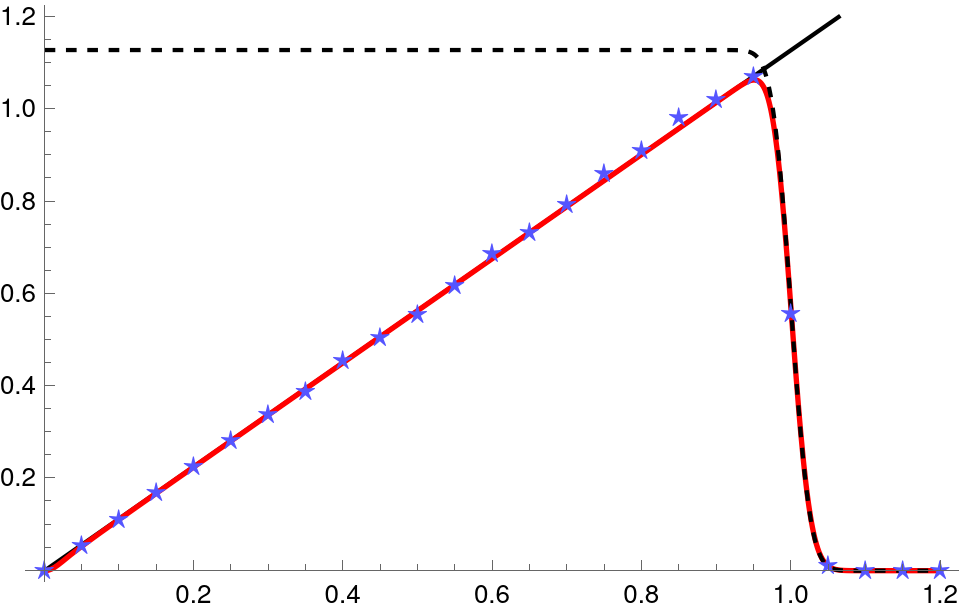}
		\end{center}
		\subcaption{$m=1$}
	\end{subfigure}
	\begin{subfigure}[h]{0.31\textwidth}
		\begin{center}
			\includegraphics[width=\textwidth]{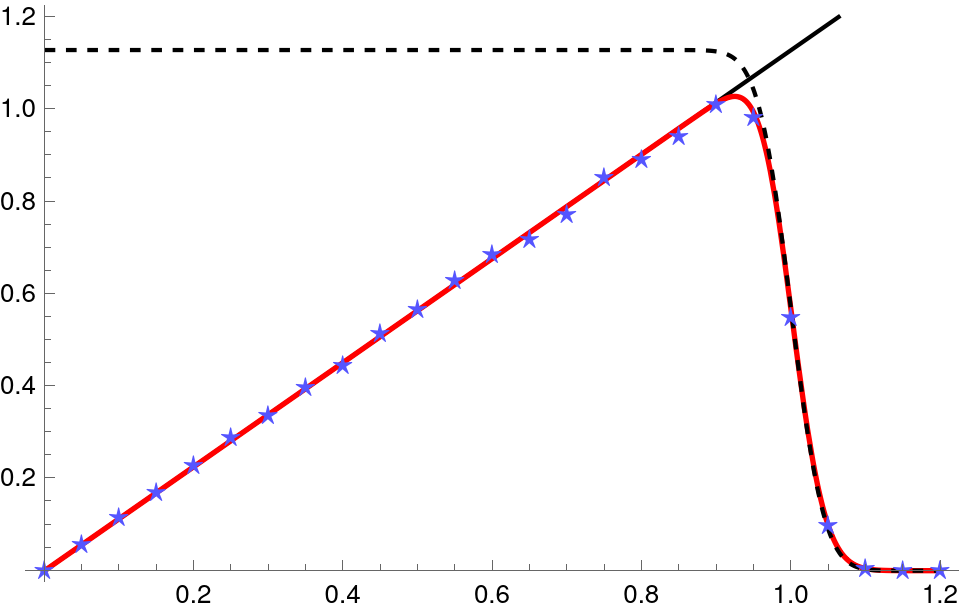}
		\end{center} \subcaption{$m=3$}
	\end{subfigure}
		\begin{subfigure}{0.31\textwidth}
		\begin{center}
			\includegraphics[width=\textwidth]{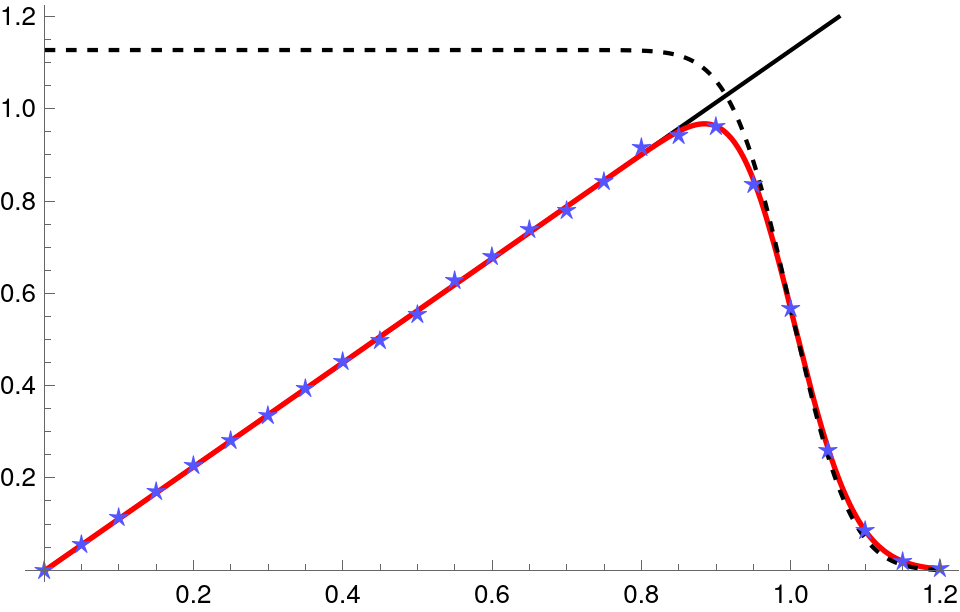}
		\end{center}
		\subcaption{$m=10$}
	\end{subfigure}

    	\begin{subfigure}{0.31\textwidth}
		\begin{center}
			\includegraphics[width=\textwidth]{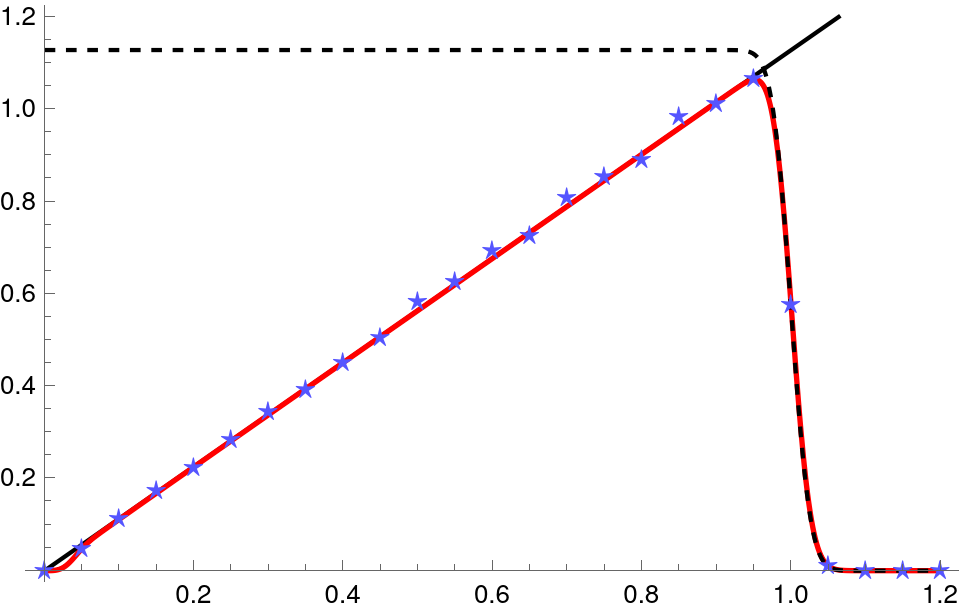}
		\end{center}
		\subcaption{$m=1$}
	\end{subfigure}
	\begin{subfigure}[h]{0.31\textwidth}
		\begin{center}
			\includegraphics[width=\textwidth]{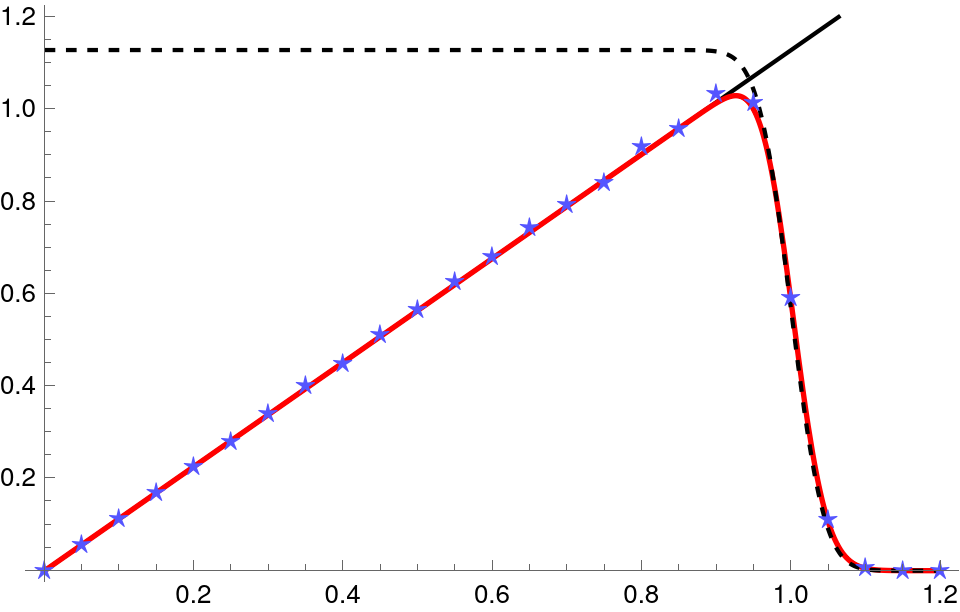}
		\end{center} \subcaption{$m=3$}
	\end{subfigure}
		\begin{subfigure}{0.31\textwidth}
		\begin{center}
			\includegraphics[width=\textwidth]{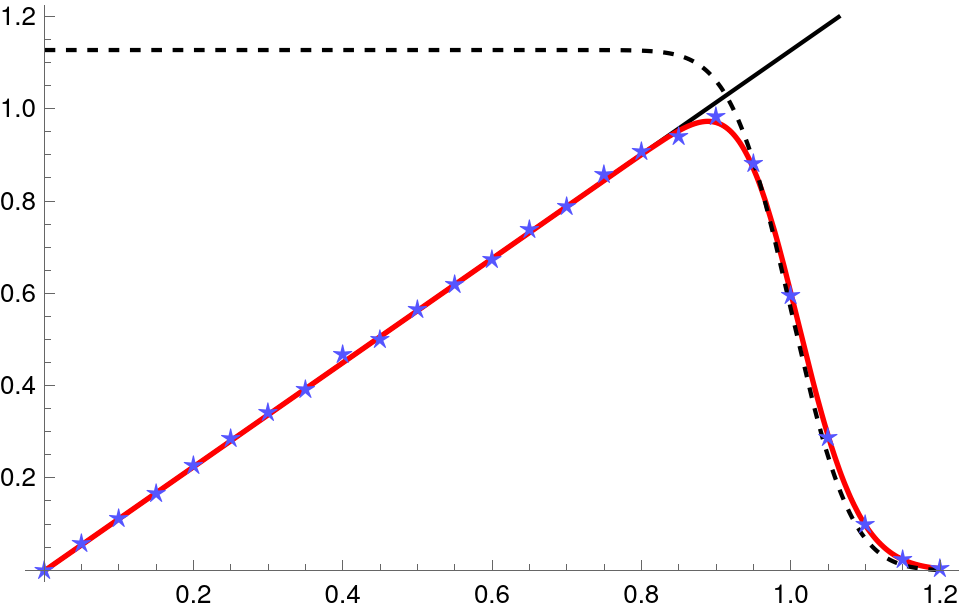}
		\end{center}
		\subcaption{$m=10$}
	\end{subfigure}
	\caption{
	The density (top line) and variance ($\beta=2$ middle and $\beta=4$ bottom line)  in the bulk and at the edge is shown for the product of $m=1,3$ and $10$ Ginibre matrices (left, middle, right column).
Figures (A)--(C) display the corresponding density \eqref{products macro density} which is singular at the origin for $m>1$. In figures (D)--(F) the variance $\frac{\beta }{ \sqrt{N \Delta Q(a)} }  V_N^{(\beta)}(a)$ is shown as a function of radius $a$ for the corresponding ensemble at $\beta=2$ with $N = 500$ (red line), and the comparison with the analytic results in Theorem~\ref{Thm_number variance}. Here, the full line corresponds to the linear prediction (i) in the bulk \eqref{VN bulk}, whereas the dotted line corresponds to the prediction (ii) for the edge \eqref{VN edge}.
The blue stars show the variance obtained via 10000 random samples. Figures (G)--(I) display the same figures at $\beta=4$, with $N = 250$.
	} \label{Fig_Vm products}
\end{figure}

The macroscopic density follows the ``Halloween hat'' law\footnote{This name was proposed by V.L. Girko.}
\begin{equation}\label{products macro density}
\frac{2}{\beta}\Delta Q^{(m)}(z) \cdot \mathbbm{1}_{\{ |z|<1 \} }= \frac{1}{m} \, |z|^{ \frac{2}{m}-2 } \cdot \mathbbm{1}_{\{ |z|<1 \} },
\end{equation}
as shown in \cite{2010PhRvE..81d1132B,2010arXiv1012.2710G,o2011products}, see Fig.\ \ref{Fig_Vm products} top line.
Note that as $N\to\infty$, we have the following large argument asymptotic \cite{fields1972asymptotic} of \eqref{Q products}: as $N\to\infty,$
\begin{equation}
Q_N^{(m)}(z) \sim \frac{\beta m}{2}|z|^{ \frac{2}{m} } - \frac{\frac{1}{m} - 1}{N} \log \abs{z},
\end{equation}
Thus asymptotically the potential \eqref{Q products} corresponds to \eqref{g ML potential} with $m=1/b$ and
$c = (1-m)/2m$.
One can also see the relation $m=1/b$ by comparing \eqref{ML macro density} and \eqref{products macro density}.

\begin{prop} \label{Prop_number variance products}
\textup{\textbf{(Expected number and variance of the product of $m$ Ginibre matrices at the origin)}}
Let $Q_N=Q_N^{(m)}$.
When rescaling $a=\TT/N^{m/2}$ with $\TT >0$, we have
\begin{align}
E^{(\beta)}( \TT)&=\lim_{N \to \infty} E_N^{(\beta)}\Big( \frac{\TT}{N^{ \frac{m}{2} }} \Big)= \sum_{j=1}^{\infty} \frac{1}{( \frac{\beta}{2}j - 1)!^{m}} \MeijerG{m , 1}{1,  m + 1}{ 1  \\  \frac{\beta }{2}\bfs{j }, 0}{\Big(\frac{\beta}{2}\Big)^m\TT^2} , \label{EN product bulk micro}
\\
\begin{split}
V^{(\beta)}( \TT) &=\lim_{N \to \infty} V_N^{(\beta)}\Big( \frac{\TT}{N^{ \frac{m}{2} }} \Big)
\\
&= \sum_{j=1}^{\infty} \frac{1}{( \frac{\beta}{2}j - 1)!^{2 m}} \MeijerG{m , 1}{1,  m + 1}{ 1  \\  \frac{\beta }{2}\bfs{j }, 0}{\Big(\frac{\beta}{2}\Big)^m\TT^2} \MeijerG{m+1 , 0}{1, m+1}{1 \\ 0, \frac{\beta}{2} \bfs{j }}{ \Big( \frac{\beta}{2} \Big)^m\TT^2 }.
\end{split}
\label{VN product bulk micro}
\end{align}
\end{prop}

Once again we denote by $\bfs{ j} = (j,\ldots,j)$ a row vector of lengths $m$.
Note that for $m=1$ we have
\begin{equation} \label{Meijer G PQ}
\frac{1}{(j-1)!} \MeijerG{1 , 1}{1, 2}{1 \\ j , 0}{x}=P(j,x), \qquad \frac{1}{(j-1)!}  \MeijerG{2 , 0}{1, 2}{1 \\ 0, j}{x}=Q(j,x).
\end{equation}
Thus one can notice that then Eq.\ \eqref{VN product bulk micro} corresponds to Eq.\ \eqref{VN ML micro} with $b=1,c=0$.

As before, we look at the small argument limit $\TT\to0$ of \eqref{EN product bulk micro} and \eqref{VN product bulk micro}.
Using Proposition~\ref{prop:Meijer G asymptotics near zero} we have
\begin{equation}
    \MeijerG{m, 1}{1, m + 1}{ 1 \\ \frac{\beta}{2}\bfs{j}, 0 }{\Big(\frac{\beta}{2}\Big)^m \TT^2}
    \sim \frac{(-1)^{m - 1}}{j \, (m - 1)!} 2^{m - 1} \Big(\frac{\beta}{2}\Big)^{m\frac{\beta}{2}j - 1} (\log \TT)^{m - 1} \TT^{\beta j}.
\end{equation}
Combining this with the identity \eqref{Meijer G sum j!}, we obtain that for small $\TT$,
\begin{equation}\label{EVproductsmallT}
    E^{(\beta)}(\TT), V^{(\beta)}(\TT) \sim \frac{(-1)^{m - 1}}{(m - 1)!} 2^{m - 1} \Big(\frac{\beta}{2}\Big)^{m\frac{\beta}{2} - 1} (\log \TT)^{m - 1} \TT^{\beta},
\end{equation}
where the leading order contribution again comes from the summand at $j=1$.
In particular, one can readily see that the asymptotic formula \eqref{EVproductsmallT} for $m=1$ agrees with \eqref{EVMLsmallT} for $b=1,c=0$.

Fig.\ \ref{Fig_Vm products} illustrates Thm.\ \ref{Thm_number variance} for products of Ginibre matrices. Because we use only the bulk scaling in this figure for the entire range of $a$, the bulk prediction given by the straight line remains the same for all plots, whereas the edge curve \eqref{VN edge} changes slightly in this scaling.
Recall that $m = 1$ corresponds to the Ginibre ensembles. The quadratic or quartic behaviour close to $a=0$
can be barely seen on this scale, compare Fig.\ \ref{Fig_Vm products} (G).
\end{eg*}

\subsection{Truncated unitary ensembles}

In this subsection we present the  truncated unitary and truncated symplectic unitary ensembles, \cite{MR1748745,khoruzhenko2021truncations} respectively. The reason is that they possess a further large-$N$ limit called weak non-unitarity, which is not covered by Theorem \ref{Thm_number variance}. It will be introduced below.
For $\beta = 2$ the ensemble is given by the joint eigenvalue distribution of the top-left corner $N \times N$  sub-matrix of a random unitary $(N + c + 1) \times (N + c + 1)$ matrix, with $c\in\mathbb{N}_+$ integer.
In the $\beta = 4$ case we consider the complex eigenvalues of a $(2 N + c + 1) \times (2 N + c + 1)$ random matrix from the compact symplectic group truncated to size $2N\times 2N$. Note that here $c$ must be odd, see \cite[Section~2]{khoruzhenko2021truncations} for details.
These truncations lead to the following potential, in contrast to the flat Haar measure without truncation:
\begin{equation} \label{Q truncated unitary}
Q_N^{\textup{trunc,w}}(z) :=
\begin{cases}
-\dfrac{c}{N} \log (1-|z|^2) & \textup{if }|z| \le 1,
\smallskip
\\
\infty & \textup{otherwise},
\end{cases} \qquad c>-1.
\end{equation}
From now on we also allow for a real valued $c$.
If $c \to -1$, this model degenerates to the circular unitary ensembles with a flat measure, where all eigenvalues lie on the unit circle, see e.g.\ \cite{forrester2010log}.

Let us first consider the \emph{weak non-unitarity limit} where $c>-1$ remains fixed as $N \to \infty$.
Here, the vast majority of the eigenvalues live on the unit circle, where the global density localises.
The potential \eqref{Q truncated unitary} is not in the general class of suitable potentials we consider above, cf. Def.\ \ref{Qsuit}. Hence the bulk and edge behaviour differ from Theorem~\ref{Thm_number variance}.
Recall that the incomplete (regularised) beta function $I_x(\alpha,b)$ is given by
\begin{equation}
\label{Idef+id}
I_x(\alpha,b)=\frac{B(x;\alpha,b)  }{ B(\alpha,b) }
=1-I_{1-x}(b,\alpha), \qquad B(x;\alpha,b)=\int_0^x t^{\alpha-1}(1-t)^{b-1}\,dt,
\end{equation}
where $B(\alpha,b)=\frac{\Gamma(\alpha)\Gamma(b)}{\Gamma(\alpha+b)}$ is the beta function, see \cite[Chapter 8]{olver2010nist}. The identity follows from simply changing variables $t\to1-t$ in the integral.

\begin{prop} \label{Prop_number variance truncated unitary}\textup{\textbf{(Number variance of truncated ensembles at weak non-unitarity)}}
Let $Q_N=Q_N^{\textup{trunc,w}}$ with $c>-1$ fixed. Then we have the following.
\begin{enumerate}[label=(\roman*)]
    \item \textup{\textbf{(Bulk)}}  For a fixed $a \in (0,1)$, we have
    \begin{equation}  \label{VN truncated weak bulk}
    \lim_{N \to \infty} V_N^{(\beta)}(a)=\sum_{j=1}^\infty I_{a^2}\Big(\frac{\beta j}{2},c+1\Big) I_{1-a^2}\Big(c+1,\frac{\beta j}{2}\Big).
    \end{equation}
    \item \textup{\textbf{(Edge)}}
    For $\SS \ge 0$, we have
    \begin{equation} \label{VN truncated weak edge}
    \lim_{N\to \infty} \frac{ \beta }{ N\beta } V_N^{(\beta)}\Big(1-\frac{\SS}{N\beta}\Big)= \frac{1}{\SS} \int_0^{\SS} P(c+1,u)Q(c+1,u)\,du.
    \end{equation}
\end{enumerate}
\end{prop}
Note that both in \eqref{VN truncated weak bulk} and \eqref{VN truncated weak edge}, the right hand sides vanish when $c \to -1.$
This is intuitively obvious since for $c=-1$ we recover the circular ensembles, where all eigenvalues are on the unit circle.

Because the density is zero at the origin, we omit the analysis there. We keep the terminology ``bulk'' limit for fixed $a$, although somewhat misleading here. This is because the spectrum is concentrated on the unit circle, and only of the order $O(1)$ eigenvalues are inside. This follows from
\begin{equation}
E^{(\beta)}(a)= \lim_{N \to \infty} E_N^{(\beta)}(a)=\sum_{j=1}^\infty I_{a^2}\Big(\frac{\beta j}{2},c+1\Big) , \qquad a \in (0,1),
\end{equation}
which is of order unity, as in the origin limit in the previous example.
This also explains why  the variance is not rescaled with a power of $N$.
It satisfies $E^{(\beta)}(a)\sim a^2/B(\beta/2,c+1)$ for small argument.
Thus from what we find in \eqref{VN truncated weak bulk}, the bulk limit at weak non-unitarity corresponds to an ``extended''  origin limit compared to the previous examples, with $a$ replaced by $\TT$.

The edge behaviour from Prop.\ \ref{Prop_number variance truncated unitary} is illustrated in Fig.\ \ref{Fig_Vm truncated weak}.

\begin{figure}[ht]
	\begin{subfigure}{0.31\textwidth}
		\begin{center}
			\includegraphics[width=\textwidth]{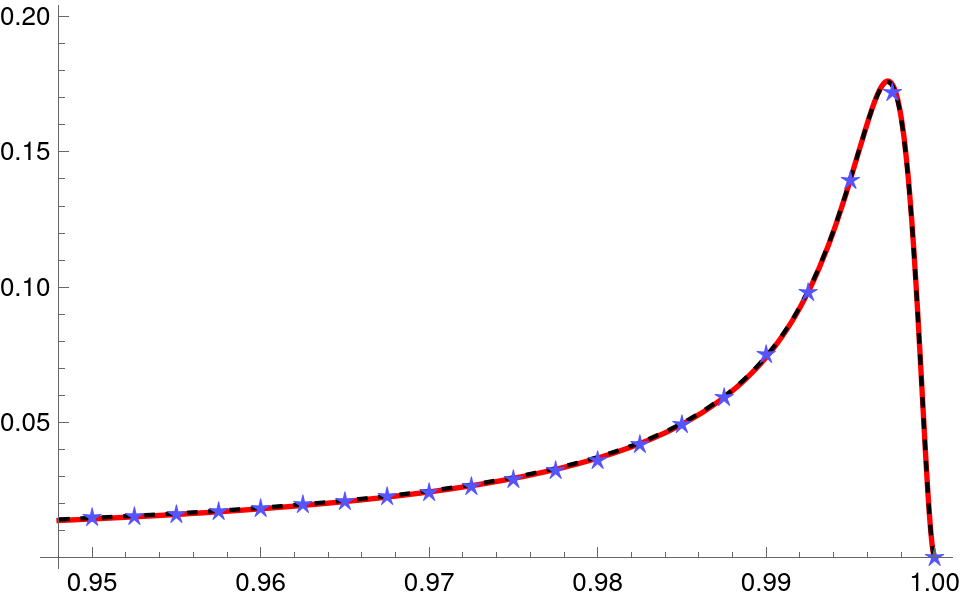}
		\end{center}
		\subcaption{$c=1$}
	\end{subfigure}
	\begin{subfigure}{0.31\textwidth}
		\begin{center}
			\includegraphics[width=\textwidth]{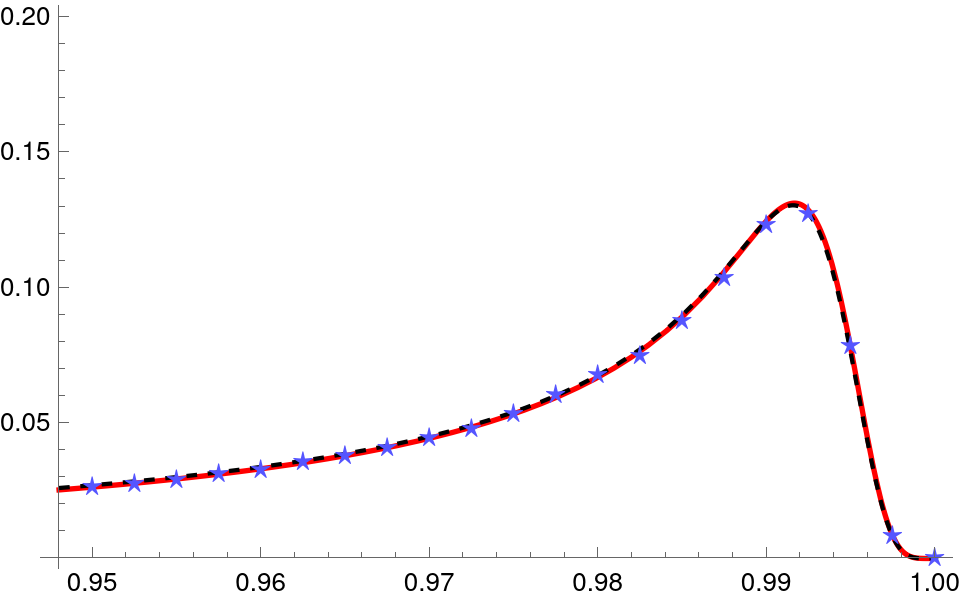}
		\end{center}
		\subcaption{$c=5$}
	\end{subfigure}
	\begin{subfigure}[h]{0.31\textwidth}
		\begin{center}
			\includegraphics[width=\textwidth]{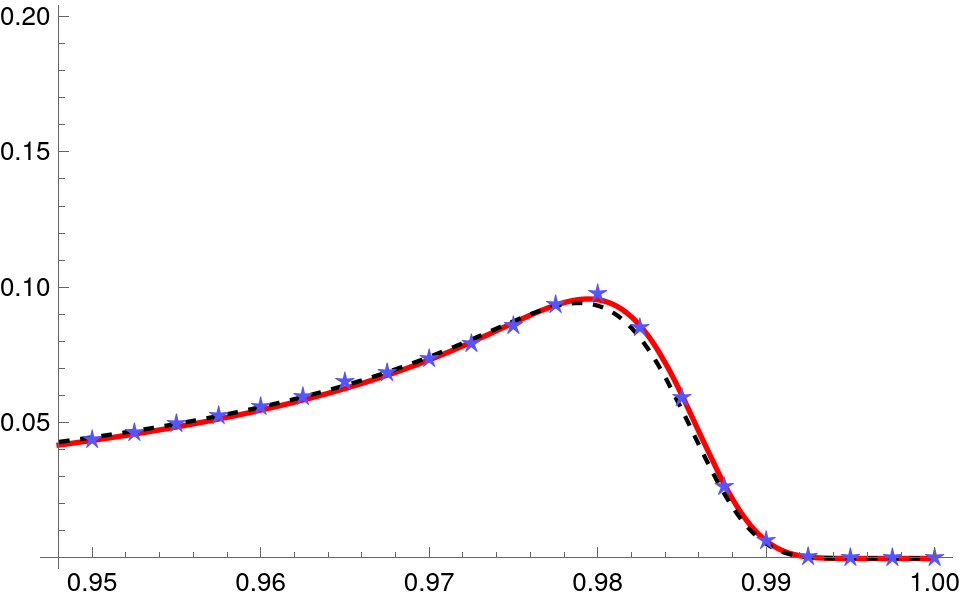}
		\end{center} \subcaption{$c=15$}
	\end{subfigure}

	\begin{subfigure}{0.31\textwidth}
		\begin{center}
			\includegraphics[width=\textwidth]{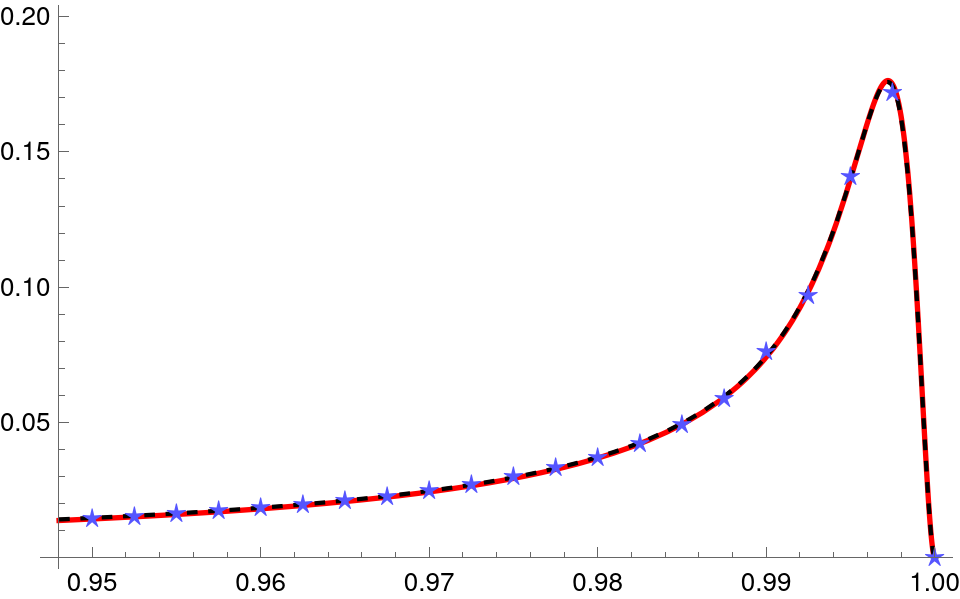}
		\end{center}
		\subcaption{$c=1$}
	\end{subfigure}
	\begin{subfigure}{0.31\textwidth}
		\begin{center}
			\includegraphics[width=\textwidth]{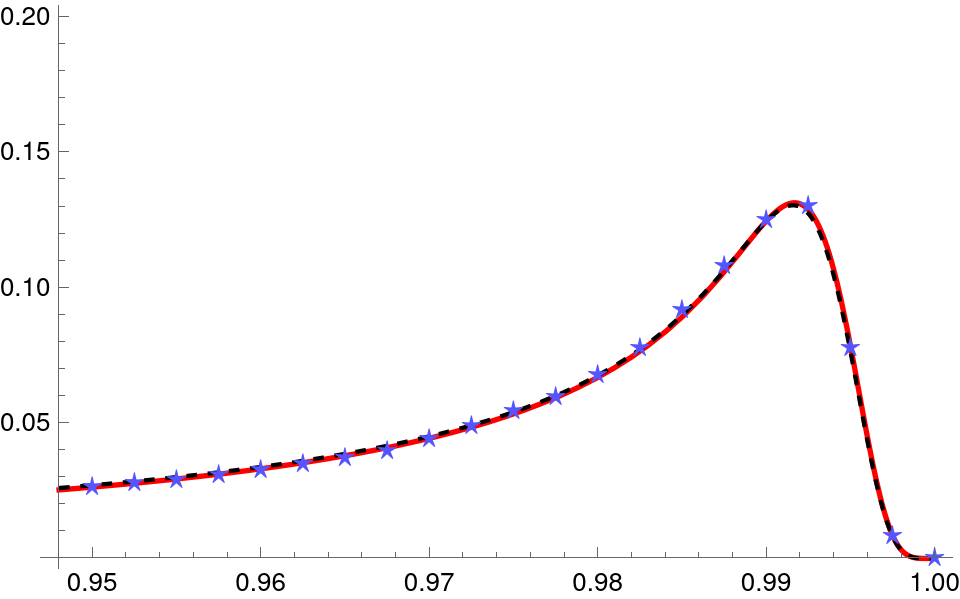}
		\end{center}
		\subcaption{$c=5$}
	\end{subfigure}
	\begin{subfigure}[h]{0.31\textwidth}
		\begin{center}
			\includegraphics[width=\textwidth]{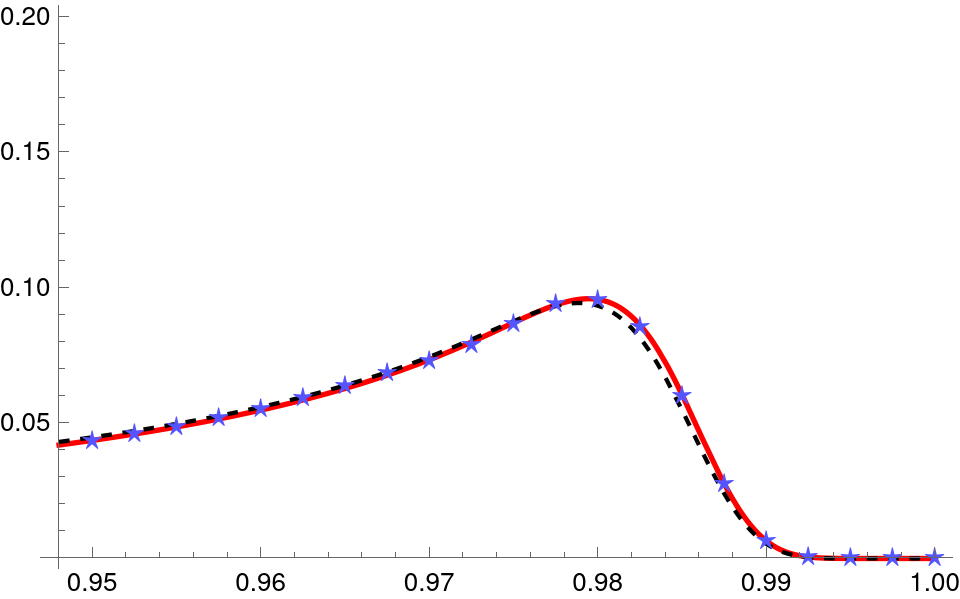}
		\end{center} \subcaption{$c=15$}
	\end{subfigure}

	\caption{ Truncated ensembles at weak non-unitarity (edge): Figures (A)--(C) in the top row display the graphs of $V_N^{(2)}(a)/N$ (red line) at $\beta=2$ and $c=1,5,15$, where $N = 500$. The comparisons with Proposition~\ref{Prop_number variance truncated unitary} (dashed black line) lie almost on top of each other. The blue stars show the variance obtained via 10000 random samples. Note that the chosen values for the radius are in the vicinity of the unit circle, $a \in [0.95, 1]$, i.e.\ here we only show the behaviour close to the edge. The figures in the bottom row (D)--(F) show the same quantities for $\beta=4$: $V_N^{(4)}(a)/N$, $N = 250$. Note that for large $c$ we begin to enter the regime of strong non-unitarity, compare Fig.\ \ref{Fig_Vm truncated strong}.
	} \label{Fig_Vm truncated weak}
\end{figure}

Let us now consider the limit at \textit{strong non-unitarity} when $c$ grows with $N$, keeping $\tilde{c}=\frac{2c}{\beta N}>0$ fixed. In the convention of \cite{MR1748745}, at $\beta=2$ this corresponds to parametrising $\mu=1/(1+\tilde{c})$ with fixed $\mu<1$ (not too close to 1). There it was shown, that the complex eigenvalues condense onto a droplet of radius $\sqrt{\mu}$.
The corresponding potential then reads
\begin{equation} \label{Q truncated unitary strong}
Q^{\textup{trunc,s}}(z):=
\begin{cases}
-\dfrac{ \beta\tilde{c} }{2} \, \log \Big( 1-\dfrac{|z|^2}{1+{ \tilde{c} } } \Big) &\textup{if } |z| \le \sqrt{1+\tilde{c}},
\smallskip
\\
\infty &\textup{otherwise},
\end{cases} \qquad \tilde{c}>0.
\end{equation}
This can be viewed as the potential \eqref{Q truncated unitary} with $c=\frac{\beta}{2}\tilde{c}N$ and fixed $\tilde{c}$, after the scaling $z \to z/\sqrt{1+\tilde{c} }$.
This rescaling is made so that the resulting droplet is the unit disc.
The potential $Q^{\textup{trunc,s}}$ is a one-parameter  generalisation of $Q^{\textup{Gin}}$, in the sense that $\lim\limits_{ \tilde{c} \to \infty } Q^{\textup{trunc,s}} = Q^{\textup{Gin}}.$
Note that the macroscopic density is given by
\begin{equation}
\label{rho strong non-unitary}
\frac{2}{\beta} \Delta Q^{\textup{trunc,s}}(z) \cdot \mathbbm{1}_{\{ |z|<1 \} } =  \frac{
\tilde{c}
(1+\tilde{c}) }{ (1+\tilde{c}-|z|^2)^2} \cdot \mathbbm{1}_{\{ |z|<1 \} }.
\end{equation}
That is unlike $Q_N^{\textup{trunc,w}}$, the potential $Q^{\textup{trunc,s}}$ is contained in the class of suitable potentials we consider above.
Hence Theorem~\ref{Thm_number variance} applies to the bulk and at the edge here. This is illustrated as another example in Fig.\ \ref{Fig_Vm truncated strong}.

\begin{figure}[ht]

\begin{subfigure}{0.31\textwidth}
		\begin{center}
			\includegraphics[width=\textwidth]{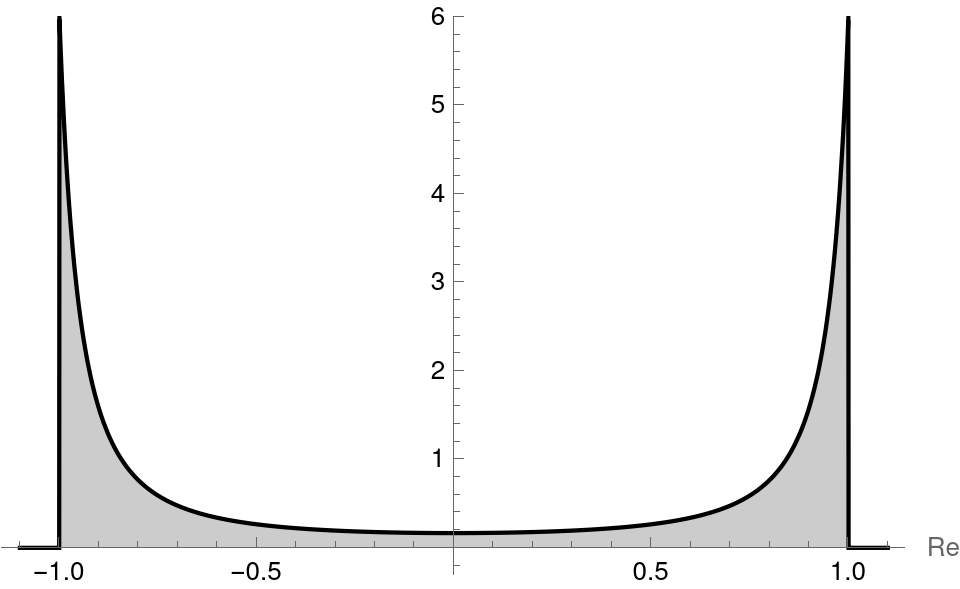}
		\end{center}
		\subcaption{$\tilde{c}=0.2$}
	\end{subfigure}
	\begin{subfigure}{0.31\textwidth}
		\begin{center}
			\includegraphics[width=\textwidth]{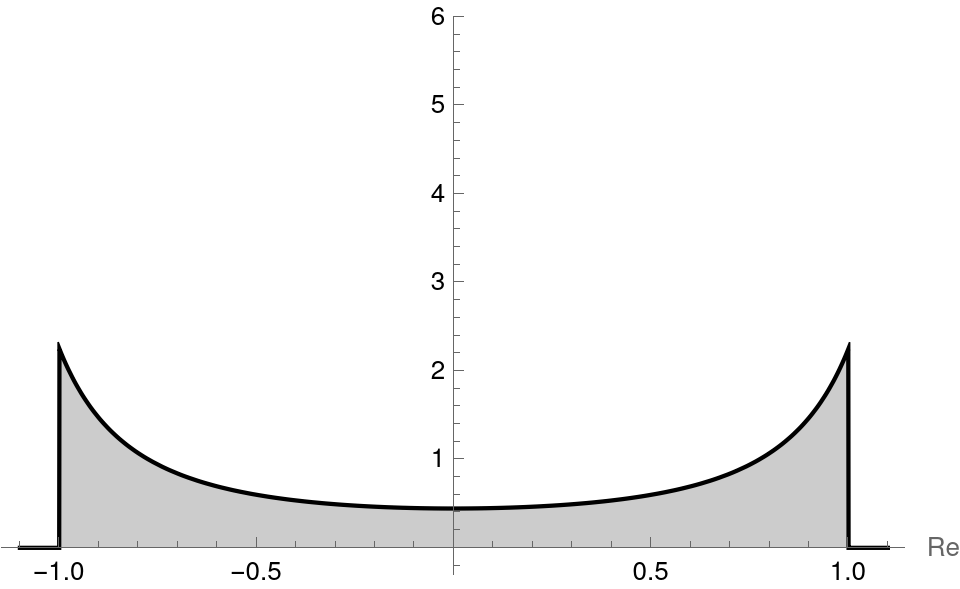}
		\end{center}
		\subcaption{$\tilde{c}=0.8$}
	\end{subfigure}
	\begin{subfigure}[h]{0.31\textwidth}
		\begin{center}
			\includegraphics[width=\textwidth]{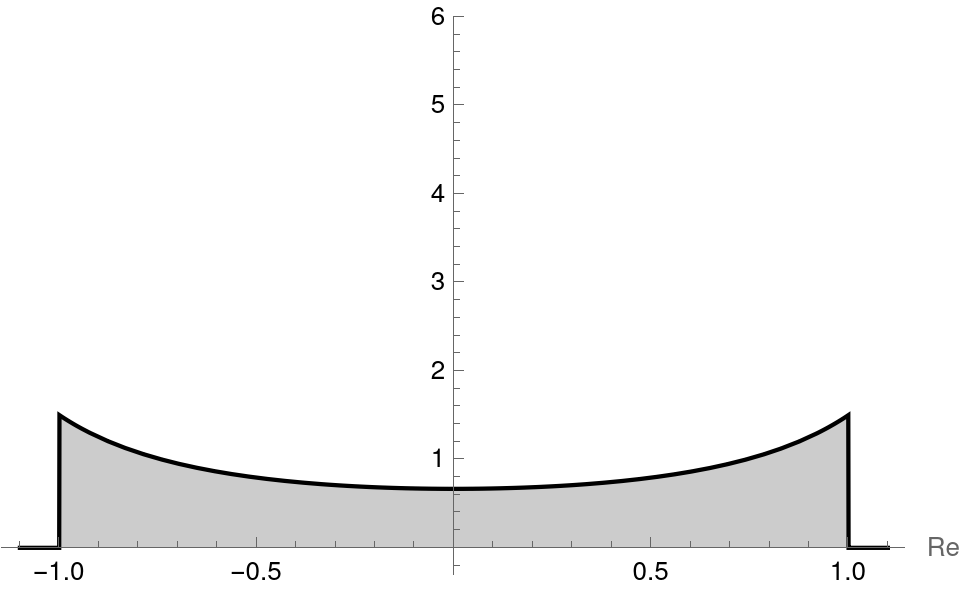}
		\end{center} \subcaption{$\tilde{c}=2.0$}
	\end{subfigure}

	\begin{subfigure}{0.31\textwidth}
		\begin{center}
			\includegraphics[width=\textwidth]{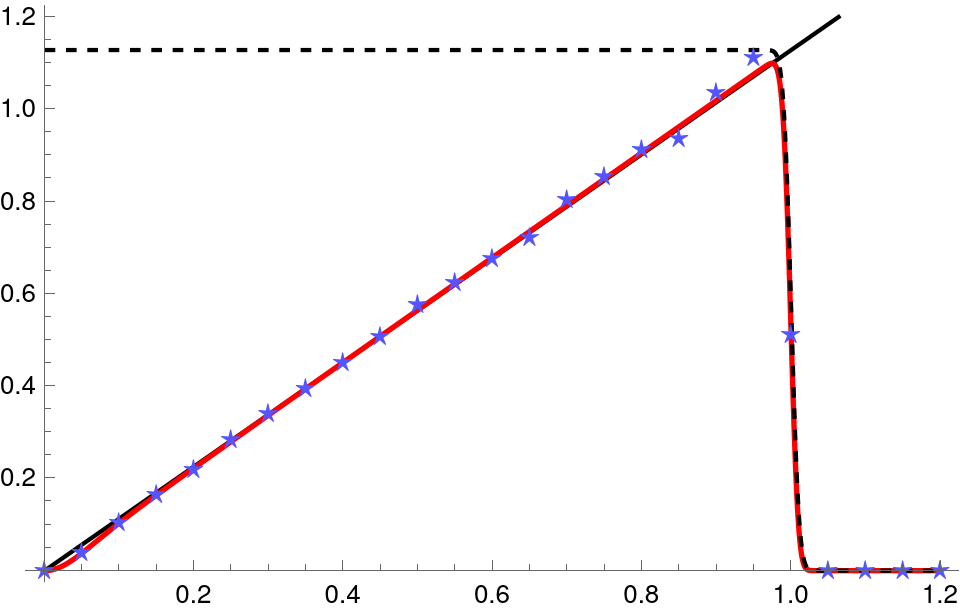}
		\end{center}
		\subcaption{$\tilde{c}=0.2$, ($c = 100$)}
	\end{subfigure}
	\begin{subfigure}{0.31\textwidth}
		\begin{center}
			\includegraphics[width=\textwidth]{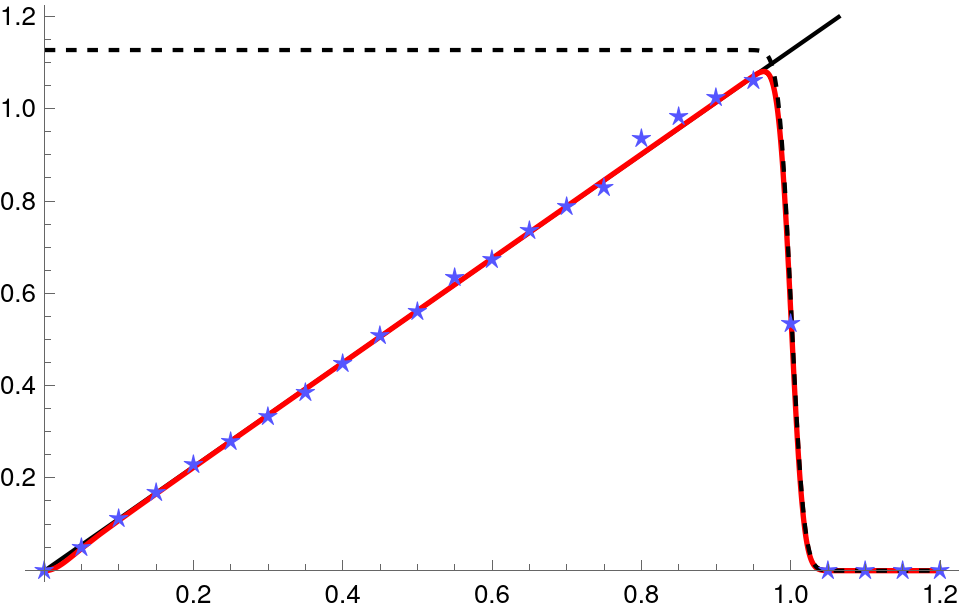}
		\end{center}
		\subcaption{$\tilde{c}=0.8$, ($c = 400$)}
	\end{subfigure}
	\begin{subfigure}[h]{0.31\textwidth}
		\begin{center}
			\includegraphics[width=\textwidth]{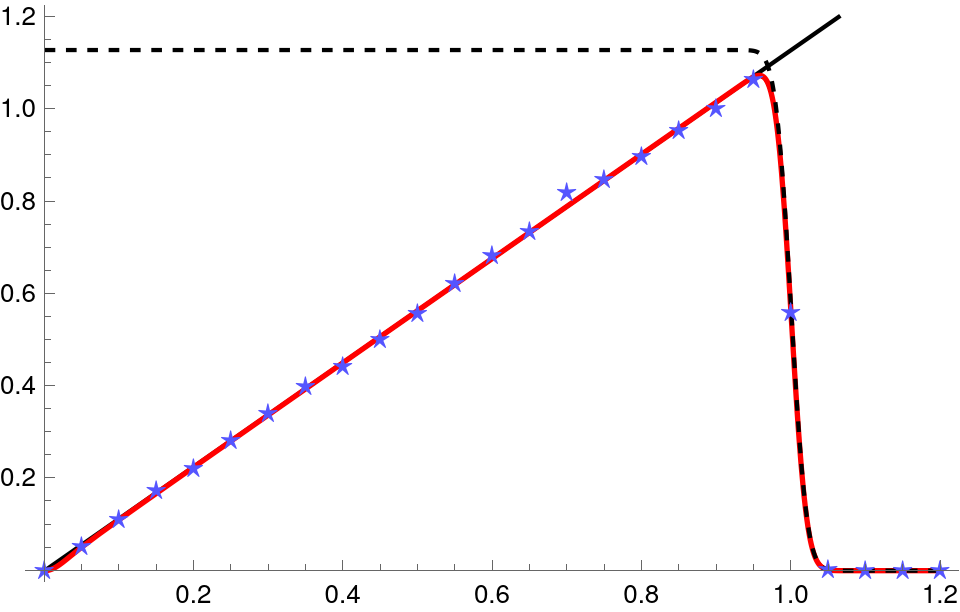}
		\end{center} \subcaption{$\tilde{c}=2.0$, ($c = 10000$)}
	\end{subfigure}

	\begin{subfigure}{0.31\textwidth}
		\begin{center}
			\includegraphics[width=\textwidth]{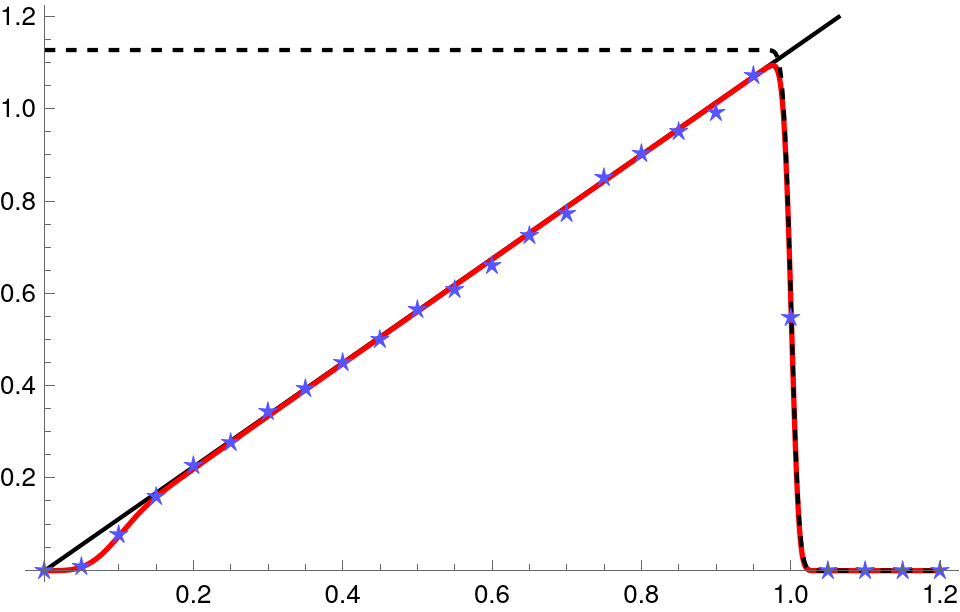}
		\end{center}
		\subcaption{$\tilde{c}=0.2$, ($c = 99$)}
	\end{subfigure}
	\begin{subfigure}{0.31\textwidth}
		\begin{center}
			\includegraphics[width=\textwidth]{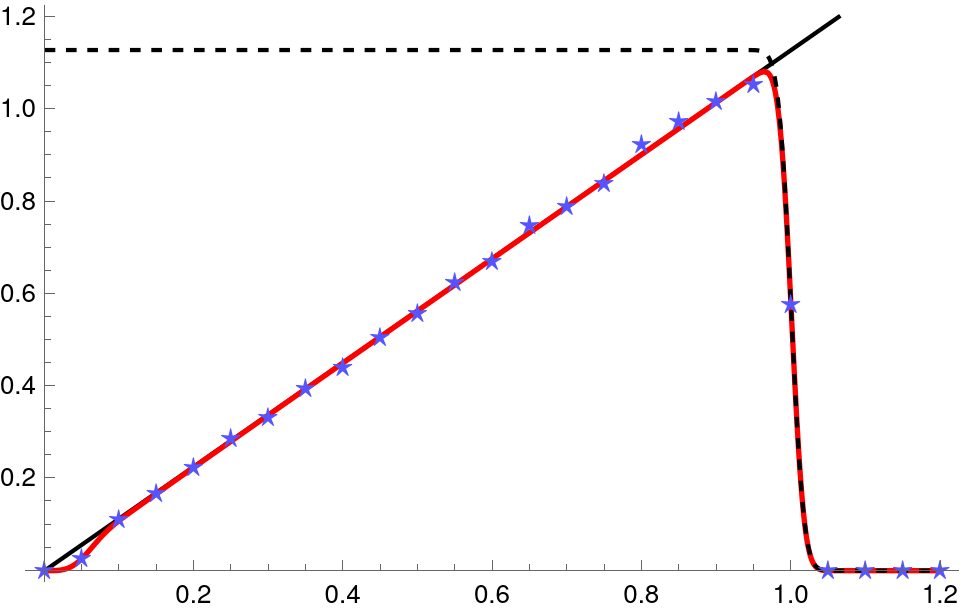}
		\end{center}
		\subcaption{$\tilde{c}=0.8$, ($c = 399$)}
	\end{subfigure}
	\begin{subfigure}[h]{0.31\textwidth}
		\begin{center}
			\includegraphics[width=\textwidth]{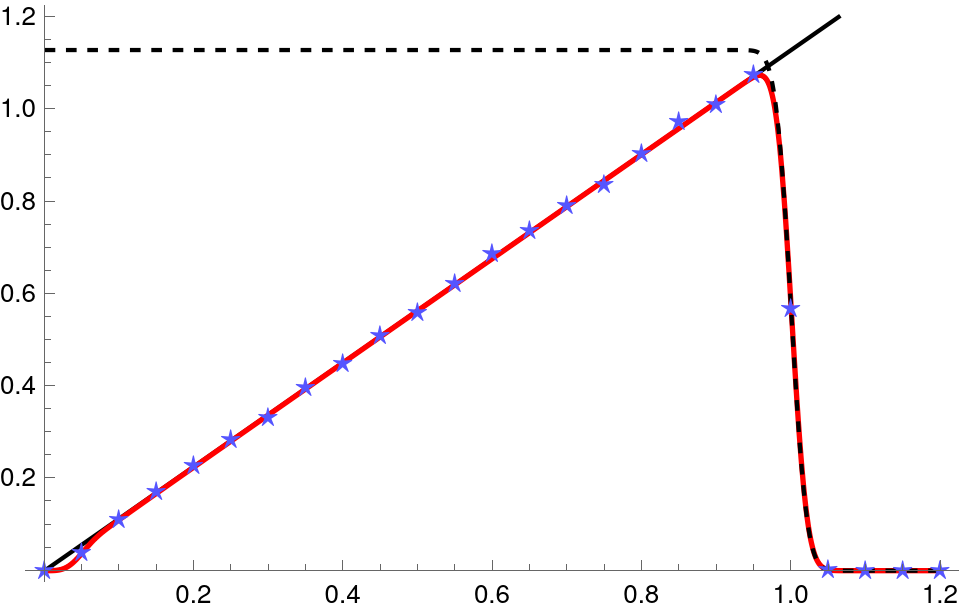}
		\end{center} \subcaption{$\tilde{c}=2.0$, ($c = 9999$)}
	\end{subfigure}

	\caption{
	Truncated ensembles at strong non-unitarity: The figures in the top row (A)--(C) show the limiting density \eqref{rho strong non-unitary}. For large values of $\tilde{c}$ we approach the flat density of the Ginibre ensemble, see Fig.\ \ref{Fig_Vm products} (A). The middle row (D)--(F)
	displays the graphs of $V_N^{(2)}(a)/\sqrt{N}$ (red line) where $N = 500$, and their comparisons with the bulk limit (full line) and edge limit (dashed line) from Theorem \ref{Thm_number variance}.  The blue stars show the variance obtained via 1000 random samples (eigenvalues rescaled by $\sqrt{1+\tilde{c}}$). The figures in the bottom row (G)--(I) show the same quantities for $V_N^{(4)}(a)/\sqrt{N}$, $N = 250$. Here, the deviation from the linear behaviour at the origin is more pronounced. Note that for $\tilde{c} \to \infty$ we recover the Ginibre case.
	} \label{Fig_Vm truncated strong}
\end{figure}

For comparison let us also consider the origin limit at strong non-unitarity.
\begin{prop} \label{Prop_number variance truncated unitary strong}\textup{\textbf{(Number variance of  truncated ensembles at strongly non-unitarity at the origin)}}
Let $Q=Q^{\textup{trunc,s}}$, $\tilde{c}>0$. When rescaling $a=\sqrt{ \frac{1+\tilde{c}}{ N\tilde{c} }} \TT$ with $\TT>0$ fixed, we have
\begin{align}
\lim_{N \to \infty} E_N\Big( \sqrt{ \frac{1+\tilde{c}}{ N\tilde{c} }}  \TT  \Big)&= \sum_{j=1}^\infty P\Big( \frac{\beta}{2} j, \frac{\beta}{2} \TT^2\Big), \label{EN trunc strong micro}
\\
\lim_{N \to \infty} V_N\Big( \sqrt{ \frac{1+\tilde{c}}{ N\tilde{c} }}  \TT  \Big)&= \sum_{j=1}^\infty P\Big( \frac{\beta}{2} j, \frac{\beta}{2} \TT^2\Big) Q\Big( \frac{\beta}{2} j, \frac{\beta}{2} \TT^2\Big).
\label{VN trunc strong micro}
\end{align}
\end{prop}
Note that the right-hand side of \eqref{EN trunc strong micro}, \eqref{VN trunc strong micro} is same as that of \eqref{EN ML micro}, \eqref{VN ML micro} when choosing $b=1$, $c=0$ there, as well as \eqref{EN product bulk micro}, \eqref{VN product bulk micro} with $m=1$. In that sense also the origin limit of the Ginibre ensembles is universal.\\

The remainder of this paper is organised as follows. In the following Section \ref{Var-finite-N} we recall the determinantal and Pfaffian structure of our planar complex and symplectic ensembles. We then prove  Proposition~\ref{Prop_VN rep} giving the expected number and variance for general rotationally symmetric potentials at finite-$N$. In Section \ref{sec:Gin} the results for the Ginibre ensembles are discussed in preparation of the general case. This includes both finite- and large-$N$ formulae for the expected number and variance for the complex and symplectic ensemble. In Section \ref{sec:Var} we prove our main Theorem~\ref{Thm_number variance} about the bulk and edge limit of the variance for general rotationally invariant potentials. The proofs for the three examples for a singular behaviour at the origin, given in Propositions \ref{Prop_number variance ML ensemble} - \ref{Prop_number variance truncated unitary strong} are presented here as well. Appendix \ref{Appendix_Meijer G} contains the derivation of a general identity for Meijer G-functions, that is needed in a special case in the proof of Proposition \ref{Prop_number variance products}, and the asymptotic at the origin.


\section{Analysis of the expected number and number variance at finite-\texorpdfstring{$N$}{N}}\label{Var-finite-N}

In this section, we give a brief exposition of the integrable structures of the ensembles we consider and derive Proposition~\ref{Prop_VN rep}.

\subsection{Correlation kernels}

By definition, the $k$-point function $R_{N,k}^{(\beta)}$ is given by
\begin{equation}\label{RNk beta}
R_{N,k}^{(\beta)}(z_1,\dots,z_k):= \frac{N!}{(N-k)!} \int_{ \C^{n-k} }    \P_N^{(\beta)} (z_1,\dots,z_N) \prod_{j=k+1}^{N} dA(z_j).
\end{equation}
When $\beta=2$ the point process \eqref{Gibbs cplx} is determinantal and
the $k$-point correlation function $R_{N,k}^{(2)}$ has the form
\begin{equation} \label{RNk det}
R_{N,k}^{(2)}(z_1,\dots,z_k)=\det \Big[ K_N(z_j,z_l) \Big]_{j,l=1}^k,
\end{equation}
where $K_N(z,w)$ is the correlation kernel of planar orthogonal polynomials.
Examples for the 1- and 2-point function are given in \eqref{ExampleR1R2beta2}.
For the rotationally invariant potentials $Q_N(z)=g_N(|z|)$ we consider here, these orthogonal polynomials are simply given by the monomials,
and \eqref{RNk det} reads \cite{Mehta}
\begin{align}
\label{KN ONP}
   K_N(z,w) & =e^{-\frac{N}{2} (Q_N(z) + Q_N(w))  } \sum_{j=0}^{N-1} \frac{(z\bar{w})^j}{h_j},
   \end{align}
where $h_j$ is given by \eqref{ortho norm}, the squared norms of these polynomials.

\begin{eg*}[Complex Ginibre ensemble] Let us consider the potential $Q^{\rm Gin}(z)=|z|^2$.
Then the correlation kernel $K_N$ of the Ginibre ensemble is given by \cite{Mehta}
\begin{equation} \label{KN Ginibre}
\begin{split}
K_N(z,w)&=N\,e^{- \frac{N}{2}(|z|^2+|w|^2) } \sum_{j=0}^{N-1} \frac{ (Nz \bar{w})^j }{ j! }.
\end{split}
\end{equation}
We refer to \cite[Section 2]{ameur2021szego} for a recent work on a detailed analysis of $K_N$.
\end{eg*}

At $\beta=4$, the point process \eqref{Gibbs symplectic} is of Pfaffian type and  $k$-point function $R_{N,k}^{(4)}$ has a Pfaffian structure \cite{MR1928853}
\begin{equation}  \label{RNk Pfa}
R_{N,k}^{(4)}(z_1,\dots, z_k) =\prod_{j=1}^{k} (\bar{z}_j-z_j)  \Pf \Big[
e^{ -\frac{N}{2}(Q(z_j)+Q(z_l)) }
\begin{pmatrix}
\kappa_N(z_j,z_l) & \kappa_N(z_j,\bar{z}_l)
\smallskip
\\
\kappa_N(\bar{z}_j,z_l) & \kappa_N(\bar{z}_j,\bar{z}_l)
\end{pmatrix}  \Big]_{ j,l=1 }^k,
\end{equation}
where $\kappa_N$ is called the (skew) pre-kernel. Examples for the 1- and 2-point function are given in \eqref{ExampleR1R2beta4}.
The skew-kernel $\kappa_N$ in \eqref{RNk Pfa} can be constructed explicitly in terms of skew-orthogonal polynomials $q_k$ and their skew-norms $s_k$ for a general rotationally invariant potential \cite{MR1928853}.
Namely, we have
\begin{equation}
q_{2k+1}(z)=z^{2k+1},\qquad q_{2k}(z)=z^{2k}+\sum_{l=0}^{k-1} z^{2l} \prod_{j=0}^{k-l-1} \frac{h_{2l+2j+2}}{h_{2l+2j+1}}, \qquad s_k=2h_{2k+1},
\end{equation}
see e.g.\ \cite{akemann2021skew,MR3066113}.
Then we have
\begin{equation}
\kappa_N(z,w)=G_N(z,w)-G_N(w,z),
\end{equation}
where
\begin{equation}
G_N(z,w) := \sum_{k=0}^{N-1} \frac{ q_{2k+1}(z)q_{2k}(w) }{ s_k }
= \frac12 \sum_{k=0}^{N-1}  \frac{ z^{2k+1} }{ h_{2k+1} }  \Big( w^{2k}+\sum_{l=0}^{k-1} w^{2l} \prod_{j=0}^{k-l-1} \frac{h_{2l+2j+2}}{h_{2l+2j+1}} \Big).
\end{equation}
Let us define
\begin{equation}
a_{2k+1}  := \frac{1}{\sqrt{2}} \frac{ h_{2k} }{ h_{2k+1} } \frac{ h_{2k-2} }{ h_{2k-1} } \cdots \frac{h_0}{h_1}, \qquad   a_{2l} :=  \frac{1}{\sqrt{2}}  \frac{h_{2l-1} }{ h_{2l} } \frac{ h_{2l-3} }{ h_{2l-2} } \cdots \frac{h_1}{h_2} \frac{1}{h_0}.
\end{equation}
Then one can write
\begin{equation} \label{G expression}
\begin{split}
G_N(z,w) = \sum_{ 0 \le l \le  k \le N-1 } a_{2k+1} a_{2l} z^{2k+1}  w^{2l}.
\end{split}
\end{equation}
It is obvious but noteworthy that
\begin{equation} \label{ak ak+1 hk+1}
a_k a_{k+1}= \frac12 \frac{1}{h_{k+1}}.
\end{equation}

\begin{eg*}[Symplectic Ginibre ensemble]
Let us consider the potential $Q^{\rm Gin}(z)=2|z|^2$. It then holds that
\begin{equation}
h_j= \frac{j!}{(2N)^{j+1}}, \qquad \prod_{j=0}^{k-l-1} \frac{h_{2l+2j+2}}{h_{2l+2j+1}}
= \frac{ k!}{ l!\,N^{k-l} } , \qquad q_{2k}(z)= \frac{ k! }{ N^k } \sum_{l=0}^{k} \frac{ N^l z^{2l}  }{ l!  }.
\end{equation}
Therefore we have
\begin{equation}
G_N(z,w)
= \sum_{ 0 \le l \le  k \le N-1 } a_{2k+1} a_{2l} z^{2k+1}  w^{2l}, \qquad a_k=(2N^3)^{\frac14}\, \frac{(2N)^{k/2}}{k!!}.
\end{equation}
\end{eg*}

\subsection{Proof of Proposition~\ref{Prop_VN rep}}

In this subsection, we determine the mean number and variance for general rotationally symmetric potentials at finite-$N$.

\begin{proof}[Proof of Proposition~\ref{Prop_VN rep}]

By definition \eqref{EN VN}, we have
\begin{equation}
 \label{EN}
 E_N^{(\beta)}(a)= \int_{ |z|<a } R_{N,1}^{(\beta)}(z)\,dA(z).
\end{equation}
It is also well known that $V_N^{(\beta)}$ is given by
\begin{equation}
\begin{split} \label{VN R12}
V_N^{(\beta)}(a)&=\int_{|z|<a} R_{N,1}^{(\beta)}(z)\,dA(z)+ \int_{|z|<a} \int_{|w|<a} R_{N,2}^{(\beta)}(z,w)-R_{N,1}^{(\beta)}(z)R_{N,1}^{(\beta)}(w)\,dA(w)\,dA(z)
\\
&= \int_{|z|<a} \int_{|w|>a} R_{N,1}^{(\beta)}(z)R_{N,1}^{(\beta)}(w)-R_{N,2}^{(\beta)}(z,w) \,dA(w)\,dA(z),
\end{split}
\end{equation}
see e.g.\ \cite[Section 16.7]{Mehta}.
For the second identity, we have used the reproducing property and normalisation,
\begin{equation}
\int_\C R_{N,2}^{(\beta)}(z,w)\,dA(w)=(N-1) R_{N,1}^{(\beta)}(z),\qquad
\int_\C R_{N,1}^{(\beta)}(w)\,dA(w)=N,
\end{equation}
which follow from the definition \eqref{RNk beta}.

We first show the $\beta=2$ case.
For the determinantal Coulomb gas ensemble \eqref{Gibbs cplx}, by \eqref{RNk det}, we have
\begin{equation}
\label{ExampleR1R2beta2}
R_{N,1}^{(2)}(z)=K_N(z,z), \qquad R_{N,2}^{(2)}(z,w)=K_N(z,z)K_N(w,w)-|K_N(z,w)|^2.
\end{equation}

By \eqref{KN ONP}, we have
\begin{equation}
\begin{split}
R_N^{(2)}(z)= e^{ -N Q_N(z) } \sum_{j=0}^{N-1} \frac{ |z|^{2j} }{h_j} , \qquad |K_N(z,w)|^2 = e^{ -N( Q_N(z)+Q_N(w) ) }  \sum_{j,k=0}^{N-1} \frac{ (z\bar{w})^j \, (\bar{z}w)^k }{  h_j \, h_k }.
\end{split}
\end{equation}
Thus the expression $E_N^{(2)}$ in \eqref{EN 2 4 rep} follows from the definition \eqref{ortho norm trunc}.

For the variance, by \eqref{KN ONP} and \eqref{ExampleR1R2beta2}, we have
\begin{equation}
\begin{split}
 V_N^{(2)}(a)&= \int_{ |z|< a } \int_{|w|>a} |K_N(z,w)|^2\,dA(w)\,dA(z)
 \\
 &=  \sum_{j,k=0}^{N-1} \frac{1}{ h_j\, h_k }\int_{ |z|<a } e^{-N  Q_N(z) } z^j \bar{z}^k \, dA(z)\int_{ |w|>a } e^{-N Q_N(w) }  w^k \bar{w}^j  \,dA(w).
\end{split}
\end{equation}
The rotationally invariant domain and potential project onto the diagonal component,
\begin{equation}
\int_{ |z|<a } e^{-N Q_N(z)} z^j \bar{z}^k \, dA(z) = \frac{1}{\pi} \int_a^\infty e^{-N g_N(r)} r^{k+j+1} \, dr \int_0^{2\pi} e^{i(j-k)\theta}\,d\theta = \delta_{j,k}\,  2 \int_0^a e^{-N g_N(r)} r^{k+j+1}  \,dr,
\end{equation}
and likewise
\begin{equation}
\int_{|w|>a} e^{-N Q_N(w) }  w^k \bar{w}^j \, dA(w) =  \frac{1}{\pi} \int_a^\infty e^{-N g_N(r) } r^{k+j+1} dr \int_0^{2\pi} e^{i(k-j)\theta}d\theta = \delta_{j,k} 2 \int_a^\infty e^{-N g_N(r)} r^{k+j+1}dr.
\end{equation}
By  \eqref{ortho norm trunc}, this completes the proof for $\beta=2$.

Next, we compute $V_N^{(4)}(a)$ for $\beta=4$.
For the planar symplectic ensemble \eqref{Gibbs symplectic}, by \eqref{RNk Pfa}, we have
\begin{align}
\label{ExampleR1R2beta4}
\begin{split}
\hspace{2em} R_{N,1}^{(4)}(z)&=(\bar{z}-z) e^{-NQ_N(z)} \kappa_N(z,\bar{z}),
\\
\hspace{2em} R_{N,2}^{(4)}(z,w)&=(\bar{z}-z)(\bar{w}-w) e^{-N(Q_N(z)+Q_N(w))} (  \kappa_N(z,\bar{z}) \kappa_N(w,\bar{w})-|\kappa_N(z,w)|^2+|\kappa_N(z,\bar{w})|^2).
\end{split}
\end{align}
Note here that
\begin{equation}
\begin{split}
R_{N,1}^{(4)}(z)& = (\bar{z}-z) e^{-NQ_N(z)} \sum_{ 0 \le l \le k \le N-1} a_{2k+1} a_{2l} \Big( z^{2k+1} \bar{z}^{2l}- \bar{z}^{2k+1} z ^{2l} \Big)
\\
&= e^{-NQ_N(z)} \sum_{ 0 \le l \le k \le N-1} a_{2k+1} a_{2l} \Big( z^{2k+1} \bar{z}^{2l+1}- \bar{z}^{2k+2} z ^{2l} - z^{2k+2} \bar{z}^{2l}+ \bar{z}^{2k+1} z ^{2l+1}\Big).
\end{split}
\end{equation}
Thus we have
\begin{equation}
\begin{split}
E_N^{(4)}(a)&= \sum_{k=0}^{N-1} \int_{|z|<a} 2a_{2k}a_{2k+1}  e^{-NQ_N(z)} |z|^{4k+2}\,dA(z) =  \sum_{k=0}^{N-1} \frac{ h_{2k+1,1}(a) }{ h_{2k+1} },
\end{split}
\end{equation}
which leads to \eqref{EN 2 4 rep}.

For the variance we have
\begin{equation}
\label{VN Pfa}
\begin{split}
 V_N^{(4)}(a) & = \int_{ |z|< a } \int_{|w|>a} (\bar{z}-z)(\bar{w}-w) e^{-N(Q_N(z)+Q_N(w))} \Big(  |\kappa_N(z,w)|^2-|\kappa_N(z,\bar{w})|^2 \Big)\,dA(w)\,dA(z)
 \\
 &= 2\int_{ |z|< a } \int_{|w|>a} (\bar{z}-z)(\bar{w}-w) e^{-N(Q_N(z)+Q_N(w))}  |\kappa_N(z,w)|^2\,dA(w)\,dA(z),
\end{split}
\end{equation}
where we have used the
symmetry $w\to\bar{w}$ under the second integral for the second identity.
Note that
\begin{equation}
\begin{split}
&\quad (\bar{z}-z)(\bar{w}-w)   |\kappa_N(z,w)|^2 = \Big( zw-z\bar{w}-\bar{z}w+\bar{z}\bar{w} \Big)  \kappa_N(z,w) \kappa_N(\bar{z},\bar{w})
\\
&= \Big( zw-z\bar{w}-\bar{z}w+\bar{z}\bar{w} \Big)    \Big( G_N(z,w)-G_N(w,z) \Big)   \Big( G_N(\bar{z},\bar{w})-G_N(\bar{w},\bar{z}) \Big).
\end{split}
\end{equation}
Notice that the expression \eqref{G expression} for $G_N(z,w)$ contains only odd powers of $z$ and even powers of $w$.
Therefore, the terms involving
\begin{equation}
G_N(z,w)G_N(\bar{z},\bar{w}), \qquad G_N(w,z)G_N(\bar{w},\bar{z})
\end{equation}
do not contribute to the integral due to the angular integration.
Furthermore, by the change of the variable $z \mapsto \bar{z}$, $w \mapsto \bar{w}$, we have
\begin{equation}
\begin{split}
V_N^{(4)}(a) & = 4\int_{ |z|< a } \int_{|w|>a}  \Big( \bar{z}w+z\bar{w} -zw-\bar{z}\bar{w} \Big) e^{-N(Q(z)+Q(w))}  G_N(z,w)G_N(\bar{w},\bar{z})\,dA(w)\,dA(z).
\end{split}
\end{equation}
To analyse this expression, let us write
\begin{align}
V_{N,\RN{1}}^{(4)}(a) & = 4\int_{ |z|< a } \int_{|w|>a}  \bar{z}w\, e^{-N(Q_N(z)+Q_N(w))}  G_N(z,w)G_N(\bar{w},\bar{z})\,dA(w)\,dA(z),
\\
V_{N,\RN{2}}^{(4)}(a) & = 4\int_{ |z|< a } \int_{|w|>a}  z\bar{w}\, e^{-N(Q_N(z)+Q_N(w))}  G_N(z,w)G_N(\bar{w},\bar{z})\,dA(w)\,dA(z),
\\
V_{N,\RN{3}}^{(4)}(a) & = 4\int_{ |z|< a } \int_{|w|>a}  zw\, e^{-N(Q_N(z)+Q_N(w))}  G_N(z,w)G_N(\bar{w},\bar{z})\,dA(w)\,dA(z),
\\
V_{N,\RN{4}}^{(4)}(a) & = 4\int_{ |z|< a } \int_{|w|>a}  \bar{z}\bar{w}\, e^{-N(Q_N(z)+Q_N(w))}  G_N(z,w)G_N(\bar{w},\bar{z})\,dA(w)\,dA(z)
\end{align}
so that
\begin{equation}
V_N^{(4)}(a)=V_{N,\RN{1}}^{(4)}(a)+V_{N,\RN{2}}^{(4)}(a)-V_{N,\RN{3}}^{(4)}(a)-V_{N,\RN{4}}^{(4)}(a).
\end{equation}
We shall show that
\begin{equation}
V_{N,\RN{1}}^{(4)}(a)=  \sum_{j=0}^{N-1} \frac{ h_{2j+1,1}(a) \, h_{2j+1,2}(a)  }{ h_{2j+1}^2 }, \qquad V_{N,\RN{2}}^{(4)}(a)=V_{N,\RN{3}}^{(4)}(a)=V_{N,\RN{4}}^{(4)}(a)=0,
\end{equation}
which completes the proof.

Since
\begin{equation}
\begin{split}
G_N(z,w) G_N(\bar{w},\bar{z}) =  \sum_{  \substack{  0 \le l \le k \le N-1  \\  0 \le l' \le k' \le N-1} }   a_{2k+1} a_{2l} a_{2k'+1} a_{2l'} \, z^{2k+1}\bar{z}^{2l'}\,  w^{2l} \bar{w}^{2k'+1},
\end{split}
\end{equation}
we have
\begin{equation}
\begin{split}
&\quad z\bar{w} \, G_N(z,w) G_N(\bar{w},\bar{z}) = \sum_{  \substack{  0 \le l \le k \le N-1  \\  0 \le l' \le k' \le N-1} }   a_{2k+1} a_{2l} a_{2k'+1} \, a_{2l'}z^{2k+2} \bar{z}^{2l'} \, w^{2l} \bar{w}^{2k'+2} .
\end{split}
\end{equation}
Notice here that there is no set of indices  that satisfies
\begin{equation}
l<k+1=l', \qquad l'<k'+1=l\ .
\end{equation}
Therefore we have $V_{N,\RN{2}}^{(4)}(a)=0.$
Similarly, we have
\begin{align}
z w \, G_N(z,w) G_N(\bar{w},\bar{z}) & =  \sum_{  \substack{  0 \le l \le k \le N-1  \\  0 \le l' \le k' \le N-1} }   a_{2k+1} a_{2l}  a_{2k'+1} a_{2l'} \, z^{2k+2}  \bar{z}^{2l'} \, \bar{w}^{2k'+1} w^{2l+1},
\\
\bar{z}\bar{w}\, G_N(z,w) G_N(\bar{w},\bar{z}) & = \sum_{  \substack{  0 \le l \le k \le N-1  \\  0 \le l' \le k' \le N-1} }   a_{2k+1} a_{2l} a_{2k'+1} a_{2l'}\, z^{2k+1}  \bar{z}^{2l'+1} \, \bar{w}^{2k'+2} w^{2l}.
\end{align}
Therefore we obtain $V_{N,\RN{3}}^{(4)}(a)=V_{N,\RN{4}}^{(4)}(a)=0$ since there is no set of indices satisfying
\begin{equation}
k+1 =l'\le k', \qquad k'=l < k +1,
\end{equation}
or
\begin{equation}
k=l' <k'+1, \qquad k'+1=l \le k.
\end{equation}

We now compute $V_{N,\RN{1}}^{(4)}(a)$.
For this, note that
\begin{equation}
\bar{z}w\, G_N(z,w) G_N(\bar{w},\bar{z}) =  \sum_{  \substack{  0 \le l \le k \le N-1  \\  0 \le l' \le k' \le N-1} }   a_{2k+1} a_{2l} a_{2k'+1} a_{2l'} \, z^{2k+1}  \bar{z}^{2l'+1} \, \bar{w}^{2k'+1} w^{2l+1}
\end{equation}
Notice here that due to angular integration only the term
\begin{equation}
k=k'=l=l'
\end{equation}
contributes to the integral.
Thus by \eqref{ak ak+1 hk+1}, we have
\begin{equation}
\begin{split}
V_{N,\RN{1}}^{(4)}(a) &= 4 \sum_{k=0}^{N-1}  ( a_{2k} a_{2k+1} )^2  \int_{|z|<a} |z|^{4k+2} \, dA(z) \int_{|w|>a} |w|^{4k+2} \,dA(w) = \sum_{k=0}^{N-1} \frac{ h_{2k+1,1}(a) \, h_{2k+1,2}(a)  }{ h_{2k+1}^2 },
\end{split}
\end{equation}
which completes the proof.
\end{proof}

\section{Ginibre ensembles revisited}\label{sec:Gin}

Before proving our main result, we discuss $E_N^{(\beta)}$ and $V_N^{(\beta)}$ for the Ginibre ensembles as an application of Proposition~\ref{Prop_VN rep}.

\subsection{Expected number}

Let $b_k(\lambda)$'s be the functions defined recursively by
\begin{equation} \label{bk functions}
b_k(\lambda)=\lambda(1-\lambda) b_{k-1}'(\lambda)+(2k-1) \lambda\,b_{k-1}(\lambda), \qquad b_0(\lambda)=1.
\end{equation}
We obtain the following.

\begin{prop} \textup{\textbf{(Expected number for Ginibre ensembles)}} \label{Prop_Ginibre EN}
Let $Q=Q^{\textup{Gin}}$. Then the following holds.
\begin{itemize}
    \item \textup{\textbf{(Finite-$N$ expression)}} For each $a>0$ and $N,$ we have
\begin{equation} \label{E_N cplx Ginibre}
\begin{split}
E_N^{(2)}(a) &=\sum_{k=0}^{N-1} P(k+1,Na^2)= Na^2+N(1-a^2) \,P(N,Na^2)-\frac{ (Na^2)^N\, e^{-Na^2} }{ (N-1)! }
\end{split}
\end{equation}
and
\begin{equation}
\begin{split}  \label{EN symplectic Ginibre}
&\quad E_N^{(4)}(a)  = \sum_{k=0}^{N-1}  P(2k+2,2Na^2) =\frac12 E_{2N}^{(2)}(a)- \frac12 \sum_{k=0}^{N-1} \frac{(2Na^2)^{2k+1} e^{-2Na^2}   }{(2k+1)!}
\\
&= N a^2 +N(1-a^2) P(2N,2Na^2) -\frac12 \frac{ (2Na^2)^{2N}\, e^{-2Na^2} }{ (2N-1)! } - \frac12 \sum_{k=0}^{N-1} \frac{(2Na^2)^{2k+1} e^{-2Na^2}   }{(2k+1)!},
\end{split}
\end{equation}
where $P(a,z)=\gamma(a,z) / \Gamma(a)$ is the (regularised) incomplete gamma function.
\smallskip
\item \textup{\textbf{(Large-$N$ expansion)}} For a fixed $a \in (0,1)$, we have \begin{equation} \label{E_N cplx Ginibre asym}
\begin{split}
E_N^{(2)}(a)  &\sim N a^2+ \frac{ (Na^2)^N\, e^{-Na^2} }{ (N-1)! } \sum_{k=1}^\infty \frac{ (-1)^k \, b_k(a^2) }{ (1-a^2)^{2k}} \frac{1}{N^{k}}
\end{split}
\end{equation}
and
\begin{equation} \label{E_N symp Ginibre asym}
\begin{split}
E_N^{(4)}(a)  &\sim N a^2+\frac12 \frac{ (2Na^2)^N\, e^{-2Na^2} }{ (2N-1)! } \sum_{k=1}^\infty \frac{ (-1)^k \, b_k(a^2) }{ (1-a^2)^{2k}} \frac{1}{(2N)^{k}} - \frac12 \sum_{k=0}^{N-1} \frac{(2Na^2)^{2k+1} e^{-2Na^2}   }{(2k+1)!},
\end{split}
\end{equation}
with $b_k$'s given by \eqref{bk functions}.
Furthermore we have when rescaling $Na^2=\TT^2$ with $\TT$ fixed
\begin{equation} \label{EN beta Ginibre micro}
E_N^{(2)}(\TT/\sqrt{ N })\sim \TT^2, \qquad  E_N^{(4)}(\TT/\sqrt{ N }) \sim \TT^2-\frac14(1-e^{-4\TT^2}).
\end{equation}
\end{itemize}

\end{prop}

See \cite{MR2536111,MR1181356} for some related statement.

Notice that by Stirling's formula, we have
\begin{equation}
\frac{ (Na^2)^N\, e^{-Na^2} }{ (N-1)! } \sim \sqrt{ \frac{N}{2\pi} } a^{2N} e^{N(1-a^2)}=  \sqrt{ \frac{N}{2\pi} }  e^{N(2\log a+1-a^2)},
\end{equation}
and we observe that  $2\log a+1-a^2<0$ for $a \in (0,1)$.

Note also that for small $\TT$, we obtain a quartic behaviour $E_N^{(4)}(\TT/\sqrt{N})\sim2\TT^4+O(\TT^6)$, in contrast to the quadratic behaviour at $\beta=2$.

\begin{proof}[Proof of Proposition~\ref{Prop_Ginibre EN}]
Let us first show the finite-$N$ expressions \eqref{E_N cplx Ginibre} and \eqref{EN symplectic Ginibre}.
Note that the first expressions
\begin{equation} \label{EN 2 4 Ginibre sum}
E_N^{(2)}(a) =\sum_{k=0}^{N-1} P(k+1,Na^2), \qquad  E_N^{(4)}(a)  = \sum_{k=0}^{N-1}  P(2k+2,2Na^2)
\end{equation}
immediately follow from \eqref{Q Gin} and \eqref{EN 2 4 rep}.

By \cite[Eq.(8.4.10)]{olver2010nist} and \eqref{KN Ginibre}, we have
\begin{equation} \label{RN Ginibre}
R_{N,1}^{(2)}(z)=N\,e^{- N|z|^2 } \sum_{j=0}^{N-1} \frac{ (N|z|^2)^j }{ j! }=N \, Q(N,N|z|^2).
\end{equation}
Using integration by parts and properties of the incomplete gamma-functions one can show that
\begin{equation}
\begin{split}
\int_0^b Q(N,x)\,dx&=b \, Q(N,b)+N-N\, Q(N+1,b) =b \, Q(N,b)+N\, P(N+1,b)
\\
&=b \, Q(N,b)+ N\,P(N,b)-\frac{ b^N \, e^{-b} }{ (N-1)! } = N+(b-N)\,Q(N,b)
-\frac{ b^n \, e^{-b} }{ (N-1)! } .
\end{split}
\end{equation}
Thus we have
\begin{equation}
\begin{split}
E_N^{(2)}(a)&=2N \int_0^a r\,Q(N,Nr^2)\,dr=\int_0^{Na^2} Q(N,x)\,dx
\\
&= N+(Na^2-N) \,Q(N,Na^2)- \frac{ (Na^2)^N\, e^{-Na^2} }{ (N-1)! },
\end{split}
\end{equation}
which gives \eqref{E_N cplx Ginibre}.
Furthermore, it follows from the recurrence relation
\cite[Eq.(8.8.5)]{olver2010nist}
of the incomplete gamma function that
\begin{equation}
\begin{split}
\sum_{k=0}^{N-1}  P(2k+2,2Na^2)- \sum_{k=0}^{N-1}  P(2k+1,2Na^2)=- e^{-2Na^2} \sum_{k=0}^{N-1} \frac{ (2Na^2)^{2k+1} }{ (2k+1)! }.
\end{split}
\end{equation}
This gives \eqref{EN symplectic Ginibre}.

We now show \eqref{E_N cplx Ginibre asym} and \eqref{E_N symp Ginibre asym}.
For this, recall that for $0<\lambda<1$
\begin{equation}
P(n, \lambda n) \sim -\frac{ (\lambda n)^n\, e^{-\lambda n} }{ (n-1)! } \sum_{k=0}^\infty \frac{ (-n)^k \, b_k(\lambda) }{ (\lambda n-n)^{2k+1} } = \frac{ (\lambda n)^n\, e^{-\lambda n} }{ n! } \sum_{k=0}^\infty \frac{ (-1)^k \, b_k(\lambda) }{ (1-\lambda)^{2k+1}\, n^{k} },
\end{equation}
see \cite[Eq.(8.11.6)]{olver2010nist}.
Using this, we have
\begin{equation}
(1-a^2) \,P(N,Na^2) \sim \frac{ (Na^2)^N\, e^{-Na^2} }{ N! } \sum_{k=0}^\infty \frac{ (-1)^k \, b_k(a^2) }{ (1-a^2)^{2k}\, N^{k} },
\end{equation}
which leads to \eqref{E_N cplx Ginibre asym}.
Moreover, \eqref{E_N symp Ginibre asym} follows from the relation between \eqref{E_N cplx Ginibre} and \eqref{EN symplectic Ginibre}.

It remains to show \eqref{EN beta Ginibre micro}.
By \eqref{EN 2 4 Ginibre sum}, we have
\begin{equation}
E_N^{(2)}(\TT/\sqrt{N}) \sim \sum_{k=0}^{\infty} P(k+1,\TT^2), \qquad E_N^{(4)}(\TT/\sqrt{N})\sim \sum_{k=0}^\infty P(2k+2,2\TT^2).
\end{equation}
Integrating the following two elementary summations term wise on the left- and right-hand side,
\begin{equation}
e^{-x} \sum_{k=0}^\infty \frac{x^k}{k!}=1, \qquad 2 e^{-2x} \sum_{k=0}^\infty \frac{(2x)^{2k+1} }{ (2k+1)! } =1-e^{-4x},
\end{equation}
we obtain the desired relations by applying the definition of the incomplete gamma function under the sum:
\begin{equation} \label{sum of P}
\sum_{k=0}^{\infty} P(k+1,x)=x,\qquad  \sum_{k=0}^\infty P(2k+2,2x)= x-\frac14(1-e^{-4x}).
\end{equation}
This gives \eqref{EN beta Ginibre micro}.
\end{proof}

\subsection{Number variance}

It is instructive to derive Theorem~\ref{Thm_number variance} for the Ginibre ensembles with the potential \eqref{Q Gin} first, as this method will be generalised to prove Theorem~\ref{Thm_number variance} for a general potential.
Note that for the complex case when $\beta=2$, it was derived in  \cite{lacroix2019rotating} but we use a slightly different computation.

\begin{proof}[Proof of Theorem~\ref{Thm_number variance} for $Q=Q^{\textup{Gin}}$]

It follows from Proposition~\ref{Prop_VN rep} that
\begin{equation} \label{VN Ginibre}
V_N^{(\beta)}(a)=  \sum_{j=1}^{N}  P\Big( \frac{\beta}{2} j,\frac{\beta}{2} Na^2\Big)\,  Q\Big( \frac{\beta}{2} j,\frac{\beta}{2} Na^2\Big).
\end{equation}
Let $X \sim \textup{Poi}( \frac{\beta}{2} Na^2)$ be the Poisson random variable with intensity $ \frac{\beta}{2} Na^2$.
Since $Q(a,z)$ is the cumulative distribution function for Poisson random variables, we have
\begin{gather}
   Q\Big( \frac{\beta}{2} j,\frac{\beta}{2} Na^2\Big)=\Prob\Big(X < \frac{\beta}{2} j\Big)= \Prob \Big(  \frac{X-\frac{\beta}{2} Na^2}{ \sqrt{ \frac{\beta}{2} N}a } <  \sqrt{ \frac{\beta}{2} } \, \frac{ j-Na^2 }{ \sqrt{N} a} \Big),
\\
 P\Big( \frac{\beta}{2} j,\frac{\beta}{2} Na^2\Big)=\Prob\Big(X \ge  \frac{\beta}{2} j\Big)= \Prob \Big(  \frac{X-\frac{\beta}{2} Na^2}{ \sqrt{ \frac{\beta}{2} N}a } \ge \sqrt{ \frac{\beta}{2} } \, \frac{ j-Na^2 }{ \sqrt{N} a} \Big) .
\end{gather}
Furthermore, by the normal approximation of the Poisson random variable, we have
\begin{equation}
 Q\Big( \frac{\beta}{2} j,\frac{\beta}{2} Na^2\Big) \sim \Prob \Big(Z < \sqrt{ \frac{\beta}{2} } \, \frac{ j-Na^2 }{ \sqrt{N} a} \Big), \qquad  P\Big( \frac{\beta}{2} j,\frac{\beta}{2} Na^2\Big) \sim \Prob \Big(Z \ge \sqrt{ \frac{\beta}{2} } \, \frac{ j-Na^2 }{ \sqrt{N} a} \Big),
\end{equation}
where $Z$ is the standard normal distribution.
Therefore if $|j-Na^2|>M \sqrt{N}$ for a large $M$, the summand in \eqref{VN Ginibre} is exponentially small, i.e.
\begin{equation}
P\Big( \frac{\beta}{2} j,\frac{\beta}{2} Na^2\Big) Q\Big( \frac{\beta}{2} j,\frac{\beta}{2} Na^2\Big)= O(e^{-c M^2}) \Big( 1- O(e^{-c M^2}) \Big) =  O(e^{-c M^2})
\end{equation}
for some $c>0.$
Thus we have
\begin{equation} \label{VN Ginibre main}
V_N^{(\beta)}(a) \sim  \sum_{j=Na^2-M\sqrt{N}}^{Na^2+M \sqrt{N}}  P\Big( \frac{\beta}{2} j,\frac{\beta}{2} Na^2\Big)\,  Q\Big( \frac{\beta}{2} j,\frac{\beta}{2} Na^2\Big).
\end{equation}

To analyse the main contribution \eqref{VN Ginibre main}, we write
$j=Na^2+2 \sqrt{ \frac{N}{\beta} } a \, t.$
Then we have
\begin{equation}
 Q\Big( \frac{\beta}{2} j,\frac{\beta}{2} Na^2\Big) \sim \Prob (Z < \sqrt{2}t) = \frac{ 1+\erf(t) }{ 2 }, \qquad
 P\Big( \frac{\beta}{2} j,\frac{\beta}{2} Na^2\Big) \sim \frac{ 1-\erf(t) }{ 2 }.
\end{equation}
For the bulk case when $ a \in (0,1)$, by the Riemann sum approximation, we obtain
\begin{equation}
\begin{split}
V_N^{(\beta)}(a) \sim 2\sqrt{ \frac{N}{\beta} } a \int_\R \frac{ 1-\erf(t)^2 }{ 4 } \,dt =  a \, \sqrt{ \frac{2N}{\beta \pi} }.
\end{split}
\end{equation}
Here we have used
\begin{equation} \label{erf int real axis}
 \int_{-\infty}^\infty \frac{ 1-\erf(t)^2 }{ 4 } \,dt = \frac{2}{\sqrt{\pi}}\int_0^\infty t \erf(t) e^{-t^2} \,dt
 =\frac{1}{\sqrt{\pi}}\int_0^\infty \erf(\sqrt{s}) e^{-s} \,ds
 =  \frac{1}{\sqrt{2\pi}},
\end{equation}
after integration by parts and upon applying \cite[Eq.(7.14.3)]{olver2010nist}. Together with $\Delta Q^{\rm Gin}(z)=\frac{\beta}{2}$ we obtain $\lim_{N\to\infty}\frac{\beta}{\sqrt{N\beta/2}}V_N^{(\beta)}(a) =2a/\sqrt{\pi}$ in the bulk.

For the near edge case, it follows from rescaling $a=1-\frac{\SS}{\sqrt{N\beta}}$ and
\begin{equation}
j= N+ \frac{2}{\sqrt{\beta}} \, (t-\SS)\sqrt{N}+O(1)
\end{equation}
that we have
\begin{equation}
\begin{split}
V_N^{(\beta)}(1-\frac{\SS}{\sqrt{N\beta}}) \sim  2\sqrt{ \frac{N}{\beta} }  \int_{-\infty}^\SS \frac{ 1-\erf(t)^2 }{ 4 } \,dt=2\sqrt{ \frac{N}{2\pi\beta} } f(\SS).
\end{split}
\end{equation}
As stated in \eqref{f(SS)}, the supplement of \cite{lacroix2019rotating} gives a second form for this integral without derivation\footnote{Notice a typo in the sign of the second term there.}. It is obtained as follows
\begin{align*}
\frac{1}{\sqrt{2\pi}}f(\SS)&= \int_{-\infty}^\SS \frac{ \erfc(t)\erfc(-t)}{ 4 } \, dt=
\frac{\SS}{4}\erfc(\SS)\erfc(-\SS)+\frac{1}{\sqrt{\pi}}\int_{-\infty}^\SS t\erf(t)e^{-t^2} \, dt
\\
&= \frac{\SS}{4}\erfc(\SS)\erfc(-\SS)-\frac{1}{2\sqrt{\pi}}\erf(\SS)e^{-\SS^2}+
\frac{1}{2\sqrt{2\pi}}(1+\erf(\sqrt{2}\SS)),
\end{align*}
upon integrating by parts twice and using the definition of the error function.
Now the proof is complete for the Ginibre ensembles.
\end{proof}

\section{Number variance for general radially symmetric potentials at large-\texorpdfstring{$N$}{N}}\label{sec:Var}

In this section, we prove our main result. For simplicity we shall drop the subscript in the potential $Q_N$ and $g_N$ here.

\subsection{Proof of Theorem~\ref{Thm_number variance}}

For the proof, we shall use the standard Laplace method to derive the asymptotic of $h_{j, 1}(a) / h_j$.
\begin{proof}[Proof of Theorem~\ref{Thm_number variance}]

We first consider the case $\beta=2$.
By Proposition~\ref{Prop_VN rep}, the case $\beta=4$ follows along the same lines; the only difference comes from the step size when applying the Riemann sum approximation.

Write
\begin{equation}
\widehat{g}_j(r):=g(r)-\frac{2j+1}{N} \log r.
\end{equation}
Thus
\begin{equation}
h_j=2 \int_0^\infty e^{-N \widehat{g}_j(r)}\,dr, \qquad h_{j,1}(a)= 2 \int_0^a e^{-N \widehat{g}_j(r)}\,dr, \qquad h_{j,2}(a)= 2 \int_a^\infty e^{-N \widehat{g}_j(r)}\,dr.
\end{equation}

Since $\Delta=\frac14(\pa_r^2+\frac{1}{r} \pa_r)$, the function $f(r):=rg_N'(r)$ satisfies $f'(r)>0$, according  to our assumptions for a suitable potential.
Therefore
\begin{equation}
r\widehat{g}_j'(r)=rg'(r)-\frac{2j+1}{N}
\end{equation}
has a unique critical point $r_j,$ i.e.
\begin{equation} \label{eq:critical rj}
r_j g'(r_j)= \frac{2j+1}{N}.
\end{equation}

Notice that
\begin{equation}
\begin{split}
\widehat{g}_j''(r) &= g''(r)+\frac{2j+1}{N} \frac{1}{r^2} = g''(r)+\frac{g'(r)}{r} -  \frac{g'(r)}{r}  + \frac{2j+1}{N} \frac{1}{r^2}
 = 4 \Delta Q(r) - \frac{1}{r^2} \Big( rg'(r)-\frac{2j+1}{N}\Big).
\end{split}
\end{equation}
Thus we have as $r\to r_j,$
\begin{equation}
\widehat{g}_j(r)=\widehat{g}_j(r_j)+2 \Delta Q(r_j)  (r-r_j)^2+O(|r-r_j|^3).
\end{equation}
Using the change of variable
\begin{equation}
r=r_j+\frac{t}{\sqrt{ N \Delta Q(r_j) }},
\end{equation}
this gives
\begin{equation}
h_j \sim \frac{ 2 \, e^{-N \widehat{g}_j(r_j)}   }{\sqrt{ N \Delta Q(r_j) }} \int_{-\infty}^\infty e^{-2t^2 }\,dt  = \sqrt{ \frac{ 2\pi }{   N \Delta Q(r_j)   } }  e^{-N \widehat{g}_j(r_j)}.
\end{equation}

Let us turn to $h_{j,1}(a)$ and $h_{j,2}(a)$.
We denote $j_* = \frac{1}{2}(N a g'(a) - 1)$.
Due to \eqref{eq:critical rj}, we have $r_{j*}=a$.
By the Laplace method, for $j$ such that $|j-j_*|> M \sqrt{N}$ for a large $M,$
\begin{equation}
\frac{ h_{j,1}(a) }{ h_j }
\frac{ h_{j,2}(a) }{ h_j }= \Big( 1-O(e^{-cM^2}) \Big) O(e^{-cM^2}).
\end{equation}
Therefore by Proposition~\ref{Prop_VN rep}, we have
\begin{equation}
V_N^{(2)}(a) \sim \sum_{ j=j_*-M\sqrt{N} }^{ j_*+M\sqrt{N} } \frac{ h_{j,1}(a)  h_{j,2}(a) }{ h_j^2 } .
\end{equation}

To analyse the main contributions, let
\begin{equation}
r_{j_*+\sqrt{N}s} \sim a+ \frac{x_s}{ \sqrt{N} }.
\end{equation}
Then we have
\begin{equation}
\begin{split}
r_{ j_*+\sqrt{N}s } g'( r_{ j_*+\sqrt{N}s } ) & = \Big(a+\frac{x_s}{\sqrt{N}}\Big)g'\Big( a+\frac{x_s}{\sqrt{N}} \Big) \sim  \Big(a+\frac{x_s}{\sqrt{N}}\Big) \Big( g'( a) +\frac{x_s}{\sqrt{N}} g''(a) \Big)
\\
&\sim a g'(a)+ \frac{1}{\sqrt{N}} (a g''(a)+g'(a)) x_s =  \frac{2 j_*+2\sqrt{N}s +1 }{ N },
\end{split}
\end{equation}
which leads to
\begin{equation}
x_s= \frac{ 2s }{ ag''(a)+g'(a) } = \frac{ s }{ 2a } \frac{1}{ \Delta Q(a)} .
\end{equation}

Using the change of variable
\begin{equation}
r=r_{j_*+\sqrt{N}s}+\frac{t}{ \sqrt{N \Delta Q(a)} },
\end{equation}
we have
\begin{equation}
\begin{split}
h_{j_*+\sqrt{N}s,1}(a) & \sim    \frac{ 2 \, e^{-N \widehat{g}_{j_*}(r_{j_*})}   }{\sqrt{ N \Delta Q(a) }} \int_{-\infty}^{ -\frac{s}{2a \sqrt{\Delta Q(a)}} } e^{-2t^2 }\,dt
\\
&=   \sqrt{ \frac{ 2\pi }{   N \Delta Q(a)   } }  e^{-N \widehat{g}_{j_*}(r_{j_*})} \frac12 \Big[ 1+\erf\Big( -\frac{s}{a \sqrt{2\Delta Q(a)}} \Big)  \Big].
\end{split}
\end{equation}
Using $h_{j,2}(a)=h_j-h_{j,1}(a)$ we have thus shown that
\begin{equation}
\frac{ h_{j_*+\sqrt{N}s,1}(a) }{ h_{j_*+\sqrt{N}s} } \sim  \frac12 \Big[ 1-\erf\Big( \frac{s}{a \sqrt{2\Delta Q(a)}} \Big)  \Big], \qquad
\frac{ h_{j_*+\sqrt{N}s,2}(a) }{ h_{j_*+\sqrt{N}s} } \sim  \frac12 \Big[ 1+\erf\Big( \frac{s}{a \sqrt{2\Delta Q(a)}} \Big)  \Big].
\end{equation}

For $a \in (0,1)$, by \eqref{erf int real axis}, we obtain
\begin{equation}
\begin{split}
V_N^{(2)}(a) & \sim  \sqrt{N} \, \frac14 \int_{-\infty}^\infty 1- \erf\Big( \frac{s}{a \sqrt{2\Delta Q(a)}} \Big) ^2 \,ds
\\
& = \sqrt{N} a \sqrt{2\Delta Q(a)} \int_{-\infty}^\infty \frac{ 1-\erf(t)^2 }{ 4 } \,dt = \sqrt{N}\, \frac{a \sqrt{ \Delta Q(a) } }{\sqrt{\pi}}.
\end{split}
\end{equation}

On the other hand, for the edge we define
\begin{equation}
a=1-\frac{ \SS }{ \sqrt{ 2 \Delta Q(1)  N} } ,
\end{equation}
and obtain
\begin{equation}
\begin{split}
V_N^{(2)}(a) & \sim  \sqrt{N} \, \frac14 \int_{-\infty}^{\sqrt{2\Delta Q(1)}\SS} 1- \erf\Big( \frac{s}{a \sqrt{2\Delta Q(a)}} \Big) ^2 \,ds
\sim \sqrt{N} a \sqrt{2\Delta Q(a)} \int_{-\infty}^{\SS} \frac{ 1-\erf(t)^2 }{ 4 } \,dt .
\end{split}
\end{equation}
This completes the proof.
\end{proof}

We now prove Corollary~\ref{Cor_in prob conv}.

\begin{proof}[Proof of Corollary~\ref{Cor_in prob conv}]
It is well known \cite{HM13,MR2934715} that the empirical measures converge to Frostman's equilibrium measure \cite{ST97}.
As a consequence, as $N \to \infty$, we have
\begin{equation}\label{RN1 beta 2 4}
\frac1{N}R_{N,1}^{(\beta)}(z) \sim  \frac{2}{\beta} \Delta Q(z) \cdot \mathbbm{1}_S(z).
\end{equation}

By \eqref{EN} and \eqref{RN1 beta 2 4}, we have that for $a \in (0,1]$,
\begin{equation} \label{EN asymp}
E_N^{(\beta)}(a) \sim N\, \frac{4}{\beta} \int_0^a r \Delta Q(r) \,dr = \frac1{\beta} \int_0^a (rg'(r))' \,dr = N \, \frac{ a g'(a) }{ \beta }
\end{equation}
as $N \to \infty$.
Now \eqref{in prob conv} follows from Theorem~\ref{Thm_number variance} and \eqref{EN asymp}.
\end{proof}

\subsection{Proof of Propositions~\ref{Prop_number variance ML ensemble} and \ref{Prop_number variance products} }

\begin{proof}[Proof of Proposition~\ref{Prop_number variance ML ensemble}]

For the potential \eqref{g ML potential} with $\beta=2$, we have
\begin{equation}
h_j=2 \int_0^\infty r^{2j+2c+1} e^{-N \frac{ r^{2b} }{ b }}\,dr =\frac{1}{b} \Big(  \frac{b}{N} \Big)^{ \frac{j+c+1}{b} } \Gamma\Big( \frac{j+c+1}{b} \Big).
\end{equation}
Similarly,
\begin{equation}
h_{j,1}(a)= \frac{1}{b} \Big(  \frac{b}{N} \Big)^{ \frac{j+c+1}{b} } \gamma\Big( \frac{j+c+1}{b}, \frac{a^{2b}}{b}N \Big),  \qquad
h_{j,2}(a)= \frac{1}{b} \Big(  \frac{b}{N} \Big)^{ \frac{j+c+1}{b} } \Gamma\Big( \frac{j+c+1}{b}, \frac{a^{2b}}{b}N \Big).
\end{equation}
Thus by Proposition~\ref{Prop_VN rep}, we have
\begin{align}
\label{EN ML}
E_N^{(2)}(a)&=\sum_{j=0}^{N-1} P\Big( \frac{j+c+1}{b}, \frac{a^{2b}}{b}N \Big),
\\
\label{VN ML}
V_N^{(2)}(a)&=\sum_{j=0}^{N-1} P\Big( \frac{j+c+1}{b}, \frac{a^{2b}}{b}N \Big)  Q\Big( \frac{j+c+1}{b}, \frac{a^{2b}}{b}N \Big).
\end{align}
Similarly, we have
\begin{align}
E_N^{(4)}(a)&=\sum_{j=0}^{N-1} P\Big( \frac{2j+c+2}{b}, \frac{2a^{2b}}{b} N \Big) ,\\
V_N^{(4)}(a)&=\sum_{j=0}^{N-1} P\Big( \frac{2j+c+2}{b}, \frac{2a^{2b}}{b} N \Big) Q\Big( \frac{2j+c+2}{b}, \frac{2a^{2b}}{b} N \Big).
\end{align}
The rescaling $a=\TT/N^{\frac{1}{2b}}$ removes the $N$-dependence inside the normalised incomplete Gamma-functions. The existence of the limiting sums can be seen for example using the quotient criterion.  Combined with the large first argument asymptotic \cite[Eq.(8.11.5)]{olver2010nist}
\begin{equation}
P(\alpha,z)\sim \frac{1}{2\pi\alpha}e^{\alpha-z}\frac{z^\alpha}{\alpha^\alpha},
\end{equation}
together with $Q(\alpha,z)=1-P(\alpha,z)$, this completes the proof.
\end{proof}

We now prove Proposition~\ref{Prop_number variance products}.

\begin{proof}[Proof of Proposition~\ref{Prop_number variance products}]

Let $\beta=2$. Note that by \eqref{Q products}, we have
\begin{equation}
e^{-NQ_N(z)}=\MeijerG{m , 0}{0,  m}{ -  \\  \bfs{0}}{ N^m |z|^2 }.
\end{equation}
It can be shown that
\begin{equation} \label{Norm products}
h_j=\int_\C |z|^{2j} e^{-NQ_N(z)}\,dA(z)= \Big( \int_\C |z|^{2j} e^{-N|z|^2} \,dA(z) \Big)^m= \Big( \frac{j!}{N^{j+1}} \Big)^m,
\end{equation}
see e.g.\ \cite[Subsection 4.1]{MR2993423}.
On the other hand, we have
\begin{align}
h_{j,1}(a)&= 2\int_0^a  \MeijerG{m , 0}{0,  m}{ -  \\  \bfs{0}}{ N^m r^2 } \,r^{2j+1}\,dr
= N^{-m(j+1)}
\int_{0}^{N^m a^2}  \MeijerG{m , 0}{0,  m}{ -  \\  \bfs{0}}{ u } \,u^{j}\,du, \\
&=  N^{-m(j+1)}  \MeijerG{m , 1}{1,  m+1}{ 1  \\  \bfs{j+1},0}{ N^ma^2 },
\end{align}
upon  change of variables $u = N^m r^2$, and the integral can be found in \cite[20.5.(1)]{BatemanV2}
This yields our first result
\begin{equation}
\frac{h_{j,1}(a)}{h_j}=\frac{1}{(j!)^m} \MeijerG{m , 1}{1,  m+1}{ 1  \\  \bfs{j+1},0}{ N^ma^2 },
\label{hj1Prod}
\end{equation}
leading to the expected number at $\beta=2$.
For $h_{j,2}(a)$ we may use Proposition~\ref{prop:Meijer G identity}, to see that
\begin{equation} \label{Meijer G sum j!}
    (j!)^m - \MeijerG{m , 1}{1,  m + 1}{ 1  \\  \bfs{j+1}, 0}{ z } = \MeijerG{m+1 , 0}{1, m+1}{1 \\ 0, \bfs{j+1}}{ z },
\end{equation}
which yields
\begin{equation}
\frac{h_{j,2}(a)}{h_j}=\frac{1}{(j!)^m} \MeijerG{m +1, 0}{1,  m+1}{ 1  \\  0,\bfs{j+1}}{ N^ma^2 },
\label{hj2Prod}
\end{equation}
In summary, this leads to the variance
\begin{equation}
V_N^{(2)}(a)= \sum_{j=1}^{N} \frac{1}{(j-1)!^{2 m}} \MeijerG{m , 1}{1,  m + 1}{ 1  \\  \bfs{j }, 0}{ N^m a^2 } \MeijerG{m+1 , 0}{1, m+1}{1 \\ 0, \bfs{j}}{ N^m a^2 }.
\end{equation}
For $\beta=4$ a simple shift $N\to2N$ and $j\to2j+1$ results into
\begin{equation}
h_{2j+1}=\Big(\frac{(2j+1)!}{(2N)^{2j+2}}\Big)^m, \qquad
h_{2j+1,1}(a) = \frac{1}{(2N)^{m(2j+2)}}\MeijerG{m , 1}{1,  m+1}{ 1  \\  \bfs{2j+2},0}{ (2N)^ma^2 },
\end{equation}
and thus to the claimed expressions for the expected number and variance for $\beta=4$.

In the microscopic regime we set $N^m a^2 = \TT^2$, to obtain
\begin{equation}
\lim_{N\to\infty }V_N^{(2)}(a)=
\sum_{j=1}^{\infty}  \frac{1}{ (j-1)!^{2m} }  \MeijerG{m , 1}{1, m+1}{1 \\ \bfs{j}, 0}{\TT^2}\MeijerG{m+1 , 0}{1, m+1}{1 \\ 0,  \bfs{2j} }{\TT^2},
\end{equation}
and likewise for the other quantities. This completes the proof.
\end{proof}

\subsection{Proof of Propositions~\ref{Prop_number variance truncated unitary} and \ref{Prop_number variance truncated unitary strong}}

We first prove Proposition~\ref{Prop_number variance truncated unitary} which is at weak non-unitarity.

\begin{proof}[Proof of Proposition~\ref{Prop_number variance truncated unitary}]

For $a \le 1$, we have
\begin{equation}
\frac{h_{j,1}(a)}{h_j}=I_{a^2}(j+1,c+1)=\Prob(X \le c),
\end{equation}
where $X \sim B(j+c+1,1-a^2).$ Here $B$ is the binomial distribution.

We present the proof for $\beta=2$. Then $\beta=4$ follows along the same lines.
Note that if $a\in (0,1)$ is fixed, we have for the bulk limit
\begin{equation}
\begin{split}
\lim_{N \to \infty} V_N^{(2)}(a) & = \sum_{j=0}^\infty I_{a^2}(j+1,c+1) \Big( 1-I_{a^2}(j+1,c+1) \Big)
\\
&= \sum_{j=0}^\infty I_{a^2}(j+1,c+1) I_{1-a^2}(c+1,j+1).
\end{split}
\end{equation}

On the other hand, for  the edge limit with rescaling
\begin{equation}
a=1-\frac{\SS}{2N}, \qquad j=t N,
\end{equation}
the binomial distribution $X$ tends to Poisson distribution with intensity $\lambda=\SS t$.
Thus we have
\begin{equation}
\Prob(X \le c) \sim Q(c+1, \SS t).
\end{equation}
Therefore we obtain
\begin{equation}
V_N^{(2)}(a) \sim N \int_0^1 P(c+1,\SS t) Q(c+1,\SS t)\,dt= \frac{N}{\SS} \int_0^{\SS} P(c+1,u)Q(c+1,u)\,du,
\end{equation}
which completes the proof.
\end{proof}

Now we prove Proposition~\ref{Prop_number variance truncated unitary strong} which is at strong non-unitarity.

\begin{proof}[Proof of Propositions~\ref{Prop_number variance truncated unitary strong}]
Note that for $Q$ given by \eqref{Q truncated unitary strong},
\begin{equation}
e^{-NQ_N(z)}= \Big( 1-\frac{|z|^2}{1+\tilde{c}}\Big)^{\tilde{c}N \frac{\beta}{2} }.
\end{equation}
Thus for $a \le \sqrt{1+\tilde{c}},$ we have
\begin{equation}
\frac{h_{j,1}(a) }{ h_j }= I_{ \frac{a^2}{1+\tilde{c}} }\Big(j+1,\frac{\beta}{2}\tilde{c}N+1\Big)=1- I_{ 1-\frac{a^2}{1+\tilde{c}} }\Big(\frac{\beta}{2}\tilde{c}N+1,j+1\Big) .
\end{equation}
Then, due to the Poisson approximation of the binomial distribution, we have that for $j$ fixed
\begin{equation}
\frac{h_{j,1}(a) h_{j,2}(a)  }{ h_j^2 } \sim P\Big( j+1, \frac{\beta}{2} \TT^2\Big) Q\Big( j+1, \frac{\beta}{2} \TT^2\Big),
\end{equation}
after parametrising  $a=\sqrt{\frac{1+\tilde{c}}{N\tilde{c}}}\TT$. Thus we obtain
\begin{equation}
V_N^{(\beta)}\Big( \sqrt{\frac{1+\tilde{c}}{N\tilde{c}}}\TT\Big) \sim \sum_{j=1}^\infty P\Big( \frac{\beta}{2} j, \frac{\beta}{2} \TT^2\Big) Q\Big( \frac{\beta}{2} j, \frac{\beta}{2} \TT^2\Big),
\end{equation}
which completes the proof.
\end{proof}

\appendix

\section{An identity for Meijer \texorpdfstring{$G$}{G}-functions} \label{Appendix_Meijer G}

The Meijer $G$-function is defined by the Mellin-Barnes integral representation
\begin{equation} \label{Meijer G}
\MeijerG {m,n} {p,q} { a_1,\dots,a_p\\
   b_1,\dots,b_q| } {z}:= \frac{1}{2\pi i} \oint_L  z^s \frac{     \prod_{j=1}^{m}\Gamma(b_j-s)   \prod_{j=1}^{n}\Gamma(1-a_j+s)      }{  \prod_{j=1+m}^{q}\Gamma(1-b_j+s) \prod_{j=1+n}^{p}\Gamma(a_j-s)       }\,ds,
\end{equation}
for $p \geq n\geq 0$ and $q \geq m \geq 0$. Here, the integration contour $L$ in $\C$ depends on the poles of the gamma functions in the numerator, see e.g.\ \cite[Chapter 16]{olver2010nist}.
Using this representation we will prove the proposition below, that expresses a particular linear combination of Meijer $G$-functions in terms of a finite series.
A special case of \eqref{eq:Meijer G identity} is then used in the proof of Proposition~\ref{Prop_number variance products}.

Recall that the integration path $L$ is chosen such that the poles of $\Gamma(b_j - s)$ ($j = 1, \dots, m$) lie to the right of $L$, and those of $\Gamma(1 - a_j + s)$ ($j = 1, \dots, n$) lie to the left, see e.g.~\cite[Section 16.17]{olver2010nist} for details.
Since $\Gamma(z)$ has (simple) poles at $z = 0, -1, -2, \dots$, here we must consider the points $s = b_j + n$ and $s = a_j - 1 - n$ for $n = 0, 1, 2, \dots$.
Furthermore, we assume that the parameters of the Meijer $G$-function satisfy $a_k - b_j \notin \mathbb{N}_{>0}$ for all $k = 1, \dots, n$ and $j = 1, \dots, m$, so that the poles of $\Gamma(b_j - s)$ do not coincide with those of $\Gamma(1 - a_j + s)$. In the proposition below there is one exception to this, for the last indices
to satisfy $a_{n+1} - b_{m+1} \in \mathbb{N}_{>0}$.

\begin{prop} \label{prop:Meijer G identity}
Let $p \geq n\geq 0$, $q \geq m\geq0$ and $\bfs{a} = \{ a_1, \dots, a_{p + 1} \}$ and $\bfs{b} = \{ b_1, \dots, b_{q + 1} \}$ such that $a_{n + 1} - b_{m + 1}=k$ is a positive integer, then
\begin{equation} \label{eq:Meijer G identity}
\begin{split}
    &\quad \MeijerG{m , n + 1}{p + 1,  q + 1}{ \bfs{a}  \\  \bfs{b} }{ z } - (-1)^{a_{n + 1} - b_{m + 1}} \MeijerG{m + 1, n}{p + 1,  q + 1}{ \bfs{a}  \\  \bfs{b} }{ z } \\
    &= \sum_{c = b_{m + 1}}^{a_{n + 1} - 1} (-1)^{a_{n + 1} - 1 - c} \, \frac{
        \prod_{l = 1}^{m} \Gamma(b_l - c) \prod_{l = 1}^{n} \Gamma(1 - a_l + c)
    }{
        \prod_{l = m + 1}^{q + 1} \Gamma(1 - b_l + c) \prod_{l = n + 1}^{p + 1} \Gamma(a_l - c)
    } \, z^c,
\end{split}
\end{equation}
where the sum rums from $c=b_{m+1}, b_{m+1}+1,\ldots,a_{n+1}-1$ in integer steps.
\end{prop}

\begin{proof}
First, note that the integrands of the two Meijer $G$-functions have the common factor
\begin{equation*}
   \gamma(z; s) := \frac{
        \prod_{l = 1}^{m} \Gamma(b_l - s) \prod_{l = 1}^{n} \Gamma(1 - a_l + s)
    }{
        \prod_{l = m + 2}^{q + 1} \Gamma(1 - b_l + s) \prod_{l = n + 2}^{p + 1} \Gamma(a_l - s)
    } \, z^s.
\end{equation*}
The factors that differ are $\Gamma(1 - a_{n + 1} + s) / \Gamma(1 - b_{m + 1} + s)$
present in  the first Meijer $G$-function $G^{m , n + 1}_{p + 1,  q + 1}$ with contour $L_1$,
and the factor $\Gamma(b_{m + 1} - s) / \Gamma(a_{n + 1} - s)$
in the second Meijer $G$-function  $G^{m + 1, n}_{p + 1,  q + 1}$ with contour $L_2$.
Due to the condition that $a_{n + 1} - b_{m + 1}=k\in \mathbb{N}_{>0}$ is an integer, it is not difficult to see that these two ratios of Gamma-functions have only a finite number of poles, located  at the very same positions,
\begin{equation}
s= b_{m + 1}, b_{m + 1} + 1, \dots, a_{n + 1} - 2, a_{n + 1} - 1,
\label{poles}
\end{equation}
because poles and zeros cancel otherwise. In fact these two ratios are proportional. To see this first notice that
\begin{equation*}
    \sin( \pi (a_{n + 1} - s) ) =  \sin( \pi (b_{m + 1} +k- s) ) = (-1)^{k
    } \sin( \pi (b_{m + 1} - s) ).
\end{equation*}
Using the relation \cite[Eq.(5.5.3)]{olver2010nist} for the Gamma function valid at non-integer argument, we thus obtain
\begin{equation} \label{eq:Gamma identity}
    \frac{ \Gamma(1 - a_{n + 1} + s) }{ \Gamma(1 - b_{m + 1} + s) }
    = (-1)^{a_{n + 1} - b_{m + 1}} \frac{ \Gamma(b_{m + 1} - s) }{ \Gamma(a_{n + 1} - s) }.
\end{equation}
Thus the two Meijer $G$-functions in \eqref{eq:Meijer G identity} have the same integrand (up to a sign).

The integrals, however, do not simply add up, because their integration paths $L_1$ and $L_2$ are slightly different: following the definition \eqref{Meijer G}, the poles \eqref{poles}
lie to the left of $L_1$ in the first function in \eqref{eq:Meijer G identity}, but to the right of $L_2$ in the second.
Because all other factors and indices agree the locations of all other poles with respect to $L_1$ and $L_2$ agree.

In order to determine the left-hand side of \eqref{eq:Meijer G identity}, we modify  $L_1$ by introducing clockwise contours around each pole \eqref{poles} and adding appropriate lines to connect with $L_{1}$.
This new path $L_{1}'$ can now be deformed continuously into $L_{2}$ without crossing any other poles.
We obtain
\begin{equation*}
\begin{split}
    &\quad \MeijerG{m , n + 1}{p + 1,  q + 1}{ \bfs{a}  \\  \bfs{b} }{ z } + (-1)^{a_{n + 1} - b_{m + 1}} \MeijerG{m + 1, n}{p + 1,  q + 1}{ \bfs{a}  \\  \bfs{b} }{ z } \\
    &= \frac{1}{2 \pi i} \int_{L_{1}} \gamma(z; s) \, \frac{\Gamma(1 - a_{n + 1} + s)}{\Gamma(1 - b_{m + 1} + s)} \, ds - (-1)^{a_{n + 1} - b_{m + 1}} \frac{1}{2 \pi i} \int_{L_{2}} \gamma(z; s) \, \frac{\Gamma(b_{m + 1} - s)}{\Gamma(a_{n + 1} - s)} \, ds \\
    &= \frac{1}{2 \pi i} \Big( \int_{L_{1}'} \gamma(z; s) \, \frac{\Gamma(1 - a_{n + 1} + s)}{\Gamma(1 - b_{m + 1} + s)} \, ds - \sum_{c = b_{m + 1}}^{a_{n + 1} - 1} \oint_{c} \gamma(z; s) \, \frac{\Gamma(1 - a_{n + 1} + s)}{\Gamma(1 - b_{m + 1} + s)} \, ds  \Big) \\
    &\quad - (-1)^{a_{n + 1} - b_{m + 1}} \frac{1}{2 \pi i} \int_{L_{2}} \gamma(z; s) \, \frac{\Gamma(b_{m + 1} - s)}{\Gamma(a_{n + 1} - s)} \, ds \\
        &= - \frac{1}{2 \pi i} \sum_{c = b_{m + 1}}^{a_{n + 1} - 1} \oint_{c} \gamma(z; s) \, \frac{ \Gamma(1 - a_{n + 1} + s) }{ \Gamma(1 - b_{m + 1} + s) } \, ds,
\end{split}
\end{equation*}
where we used \eqref{eq:Gamma identity} to cancel the two integrals over $L_2$.
The remaining sum is evaluated with the residue theorem.
Note that the contours around the poles at $s=c$ are clockwise, and that all factors except $\Gamma(1 - a_{n + 1} + s)$, which  has a simple pole at $c$, are smooth near $s = c$.

Furthermore the residue of $\Gamma(z)$ at $z = - n$ is $(-1)^n / \Gamma(n + 1)$, see \cite[5.1(i)]{olver2010nist}.
Thus for each summand we have
\begin{equation*}
\begin{split}
    &\quad -\frac{1}{2 \pi i} \oint_{c} \gamma(z; s) \, \frac{ \Gamma(1 - a_{n + 1} + s) }{ \Gamma(1 - b_{m + 1} + s) } \, ds
    = \gamma(z; c) \, \frac{\Res_{s = c}(\Gamma(1 - a_{n + 1} + s))}{\Gamma(1 - b_{m + 1} + c)} \\
    &= \frac{
        \prod_{l = 1}^{m} \Gamma(b_l - c) \prod_{l = 1}^{n} \Gamma(1 - a_l + c)
    }{
        \prod_{l = m + 2}^{q + 1} \Gamma(1 - b_l + c) \prod_{l = n + 2}^{p + 1} \Gamma(a_l - c)
    } \ \frac{(-1)^{a_{n + 1} - 1 - c}}{\Gamma(1 - b_{m + 1} + c) \Gamma(a_{n + 1} - c)} \, z^c
\end{split}
\end{equation*}
and merging the fractions by adjusting the indices in the denominator yields the summands in \eqref{eq:Meijer G identity}.
\end{proof}

\begin{rmks*}
We give some remarks on special cases of Proposition~\ref{prop:Meijer G identity}.
\begin{enumerate}
    \item If $a_{n + 1} - b_{m + 1}$ is a negative integer or zero, then \eqref{eq:Meijer G identity} still holds with the convention that the empty sum is zero.
    In this case no surgery of the integration path is necessary and it is clear from \eqref{eq:Gamma identity} that the end result must be zero.
    This case is known in the literature (e.g.~\cite[Section 5.3.2, Eq.~(7)]{luke1969special}).

    \item In the proof of Proposition~\ref{Prop_number variance products} we need the case $b_{m + 1} = 0$, $a_{n + 1} = 1$.
    Here the sum on the right-hand side consists of a single term with $c = 0$ and thus it is independent of $z$:
    \begin{equation}
        \MeijerG{m , n + 1}{p + 1,  q + 1}{ \bfs{a}  \\  \bfs{b}}{ z } + \MeijerG{m + 1, n}{p + 1,  q + 1}{ \bfs{a}  \\  \bfs{b} }{ z }
    = \frac{
        \prod_{l = 1}^{m} \Gamma(b_l) \prod_{l = 1}^{n} \Gamma(1 - a_l)
    }{
        \prod_{l = m + 1}^{q+1} \Gamma(1 - b_l) \prod_{l = n + 1}^{p+1} \Gamma(a_l)
    }.
    \end{equation}
    \item In general, when $a_{n+1}>b_{m+1}\geq0$ are integers, the right hand side of \eqref{eq:Meijer G identity} becomes a polynomial.
    \end{enumerate}
\end{rmks*}

For the small argument limit $\TT \to 0$ of the results in Proposition~\ref{Prop_number variance products} we need the following small argument asymptotic of the Meijer $G$-function.
\begin{prop}\label{prop:Meijer G asymptotics near zero}
    Let $m \geq 0$ and $j \geq 1$. Then as $z \to 0$ we have
    \begin{equation*}
        \MeijerG{m, 1}{1, m+1}{ 1 \\ \bfs{j}, 0 }{z}
        = \frac{(-1)^{m - 1}}{j \, (m - 1)!} (\log z)^{m - 1} z^j + \begin{cases}
            O(z^{j + 1}), & m = 1, \\
            O( (\log z)^{m - 2} z^j), & m \geq 2.
        \end{cases}
    \end{equation*}
\end{prop}

\begin{proof}
By the similar computation in the proof of Proposition~\ref{prop:Meijer G identity}, using the residue theorem, we have
\begin{equation} \label{eq:MeijerG series}
    \MeijerG{m, 1}{1, m+1}{ 1 \\ \bfs{j}, 0 }{z}
    = - \sum_{c = j}^{\infty} \Res_{s = c}\Big( \Gamma(j - s)^m \frac{z^s}{s} \Big),
\end{equation}
where $c \in \{j, j+1, \dots \}$ are the poles of the function $\Gamma(j - s)$.
It then follows from the generalised Leibniz rule for derivatives that
\begin{equation*}
    \Res_{s = c} \Big( \Gamma(j - s)^m \frac{z^s}{s} \Big)
    = \sum_{k = 0}^{m - 1} \sum_{l = 0}^{k} \sum_{r = 0}^{l} \binom{m}{k + 1} \frac{(-1)^{l - r}}{(k - l)! \, r!} \lim_{s \to c} \Big[ (s - c)^{k + 1} \partial_{s}^{k - l} \Gamma(j - s)^m \Big] \frac{(\log z)^r z^c}{c^{1 + l - r}}.
\end{equation*}
As $z \to 0$, only the summand with $r = l = k = m - 1$ contributes to the leading term. Thus we obtain that as $z \to 0$,
\begin{equation*}
    \Res_{s = c} \Big( \Gamma(j - s)^m \frac{z^s}{s} \Big)
    = \frac{1}{(m - 1)!} \Big( \frac{(-1)^{c - j + 1}}{(c - j)!} \Big)^m \frac{(\log z)^{m - 1} z^c}{c} + \begin{cases}
        0, & m = 1, \\
        O((\log z)^{m - 2} z^c), & m \geq 2.
    \end{cases}
\end{equation*}
Furthermore, the leading term in \eqref{eq:MeijerG series} comes from the $c = j$ summand.
Simplifying the resulting expression leads to the claimed asymptotic formula.
Note that in the $m = 1$ case the first subleading term is given by the $c = j + 1$ summand.
\end{proof}

\bibliographystyle{abbrv}
\bibliography{RMTbib}

\noindent\emph{Data availability:}
The numerical data generated for part of the figures are merely an illustration of our entirely theoretical findings. These data are available upon request from the authors.

\noindent\emph{Conflict of interest:}
The authors have not disclosed any competing interests.
\end{document}